\documentclass[10pt]{amsart}
\usepackage[T1]{fontenc}

\usepackage{tgbonum}
\usepackage{amsthm,amsmath}
\usepackage[numbers]{natbib}
\usepackage[colorlinks,citecolor=blue,urlcolor=blue]{hyperref}
\usepackage{cleveref}
\usepackage[toc,page,titletoc]{appendix}
\usepackage[british]{babel}
\usepackage{nameref}
\usepackage{amssymb}
\usepackage{amsfonts}
\usepackage{mathrsfs}
\usepackage{amstext}
\usepackage{color}
\usepackage[version=4]{mhchem}
\usepackage{graphicx}
\usepackage{url}
\usepackage[all]{xy}
\usepackage{lmodern}
\usepackage{extarrows}
\usepackage{bbm}
\usepackage{mathtools}

\numberwithin{equation}{section}

\theoremstyle{plain}
\newtheorem{theorem}{Theorem}[section]

\newtheorem{corollary}[theorem]{Corollary}
\newtheorem{proposition}[theorem]{Proposition}

\theoremstyle{definition}
\newtheorem{definition}[theorem]{Definition}
\newtheorem{example}[theorem]{Example}
\theoremstyle{remark}
\newtheorem{remark}[theorem]{Remark}

\newcommand{\cC}{\mathcal{C}}
\newcommand{\cL}{\mathcal{L}}
\newcommand{\cR}{\mathcal{R}}
\newcommand{\R}{\mathbb{R}}

\allowdisplaybreaks

\ExplSyntaxOn
\keys_define:nn { mhchem }
 {
  arrow-min-length .code:n =
   \cs_set:Npn \__mhchem_arrow_options_minLength:n { {#1} } 
 }
\ExplSyntaxOff
\mhchemoptions{arrow-min-length=1.5 em}

\begin{document}
\title[On the Regularity of  Reaction Systems]
  {On the Regularity of  Reaction Systems}
\author{Chuang Xu}
\address{Department of Mathematics\\
University of Hawai'i at M\={a}noa, Honolulu, Hawai'i\\
96822, USA
}
\email{chuangxu@hawaii.edu}

\subjclass[2020]{}

\date{\today}

\noindent

\begin{abstract}
Reaction networks have been widely used as generic models in diverse areas of applied sciences, such as biology, chemistry, ecology, epidemiology, and computer science.  A reaction network incorporating noisy effects is modeled as a continuous time Markov chain (CTMC) and is called a \emph{stochastic reaction system}. In contrast, the mean field limit of a sequence of volume-scaled stochastic reaction systems as the volume tends to infinity is modeled as an ordinary differential equation (ODE) and is called a \emph{deterministic reaction system}. \emph{Non-explosivity} of CTMCs and global existence of solutions of ODEs capture the regularity of respective dynamical processes. In this paper, we study the regularity of reaction systems, in both stochastic and deterministic senses. By constructing a simple linear Lyapunov function, we obtain the regularity in both senses for a class of   
reaction systems in terms of a simple checkable condition. In particular, we prove that (a) every binary mass-action system is regular, and (b) every second-order endotactic  mass-action system is regular. We apply our results to diverse models in biochemistry, epidemiology, ecology, synthetic biology, and natural computing  in the literature.
\end{abstract}


\maketitle

\section{Introduction}

Reaction networks are widely used in modeling diverse  science phenomena \cite{F19}. The original motivation seems to come from biochemistry, and hence the mathematical theory of reaction networks is conventionally referred to as Chemical Reaction Network Theory (CRNT) \cite{F19}. When the molecule counts of species are low, reaction networks are modeled as stochastic systems in terms of continuous time Markov chains (CTMCs) \cite{AK11}.  In contrast, when species are abundant \cite{AK11}, reaction networks are modeled as deterministic systems in terms of ordinary differential equations (ODEs), as mean field approximations of stochastic systems as the \emph{volume} of the reaction fluid tends to infinity.

In CRNT, a particular focus is given to \emph{bimolecular}  reaction networks where \emph{molecularity} (the total  molecule counts of each reactant) of each reaction  is no more than two \cite{PKC12,ACK20}. This is because,  {for instance, in biochemistry,  it is in practice} less likely for more than two molecules to bind together by collision to form a complex as the reaction may need much more activation energy to proceed over a chemical potential well   \cite[Chapter~14]{C06}. For that reason, for any reaction, a product with more than two molecules may be assumed to first dissociate spontaneously thanks to a large amount of internal energy before stabilizing by collision. Based on this rationale,  reactions of bigger molecularities are usually interpreted as a sequence of bimolecular reactions that proceed in multiple steps \cite{FMKT06}, in analogy to {the binary representation of numbers}.  This analogy might explain why the term ``bimolecular reaction'' may be used interchangeably with \emph{``binary reaction''} (both reactant and product have no more than two total molecule counts) in the literature of CRNT \cite{AK18}.

Endotacticity is a property for embedded graphs in Euclidean spaces, which was originally introduced in \cite{CNP13}. The property is designed to study \emph{permanence} of reaction networks in the deterministic regime, which means \emph{the associated ODE has a compact global attractor in the interior of the its invariant subspace}. It is believed that endotacticity of the reaction graph implies \emph{dissipativity} of the associated reaction system,  i.e., \emph{the energy of the system dissipates as the system evolves} \cite{W72}. {Although}  endotacticity and strong endotacticity (see Section~\ref{sect:Endotactic_system} for their definitions) have been shown sufficient for the permanence of deterministic mass-action systems  \cite{P12,CNP13,GMS14,X26} in two dimensions and in arbitrarily finite dimensions, respectively, they are in general \emph{insufficient} for the regularity in the \emph{stochastic} regime, as evidenced by counterexamples in terms of two-species stochastic mass-action systems \cite{ACKN20}. Hence it is interesting to know if other factors  of a reaction network (e.g., the \emph{order}\---the maximal molecularity of all reactions for a stochastic mass-action system)  might jointly play a decisive role with the endotacticity in determining the regularity of an endotactic mass-action system (particularly in the stochastic regime).

In this paper, using a simple linear Lyapunov function\footnote{In the stochastic sense, it is a function $V$ satisfying the condition in Proposition~\ref{prop:regularity}.} in both deterministic and stochastic senses, 
we show that a wide class of  reaction systems are regular in the senses that the stochastic system is non-explosive (Theorem~\ref{thm:regularity}) and the deterministic system has a global solution (Theorem~\ref{thm:regularity-ODE}). Examples of these regular reaction systems include (a) \emph{binary mass-action systems} and (b)  \emph{second-order} endotactic mass-action systems (Theorem~\ref{thm:endotactic-regularity}, as well as other systems which are non-increasing in total molecule counts (Theorem~\ref{thm:1st-order-regularity}(C1)), have a decreasing species (Theorem~\ref{thm:1st-order-regularity}(C2)), or sub-conservative (Theorem~\ref{thm:1st-order-regularity}(C3)), up to a unimolecular sub reaction network. 

\section{Preliminaries}\label{sect:preliminaries}

\subsection*{Notation}
Let $\mathbb{R}$ be the set of real numbers and $\mathbb{R}_+$ the set of non-negative real numbers. Let $\mathbb{Z}$, $\mathbb{N}_0$, and $\mathbb{N}$ denote the set of integers, non-negative integers, and positive integers, respectively. Let $\{e_i\}_{i=1}^d$ be the standard basis of $\R^d$ and $\mathbf{1}_d=\sum_{i=1}^de_i$. We simply write $\mathbf{1}$ short for $\mathbf{1}_d$ when $d$ is deemphasized and clear from the context. We denote $\|x\|_1$ the $\ell_1$-norm of a vector $x\in\R^d$. For $x,y\in\mathbb{N}^d_0$,  $x^{\underline{y}}=\prod_{i=1}^dx_i^{\underline{y_i}}$ for $x\ge y$ (coordinate-wise);  $x^{\underline{y}}=0$ otherwise, and $a^{\underline{b}}=\frac{a!}{(a-b)!}$ is the falling factorial for $a,b\in\mathbb{N}_0$ and $a\ge b$. Similarly, for $x,y\in\mathbb{R}^d_+$, $x^{y}=\prod_{i=1}^dx_i^{y_i}$ for all $x\in\mathbb{R}^d_+$ with the convention $0^0=1$.

\subsection*{Reaction networks}
A \emph{reaction network} is a triple of sets: a set of \emph{species} $\mathcal{S}$, a set of \emph{complexes} $\cC$, and a set of \emph{reactions} $\cR$. Species are barely symbols, which are abstracted from chemical species. Each complex is a linear combination of species with non-negative integer coefficients (called \emph{stoichiometric coefficients}), and $\cR$ is a relation on $\cC$ with each reaction denoted $y\ce{->}y'$, where the complexes $y$ (called the \emph{reactant}) and $y'$ (the \emph{product}) are \emph{distinct}, and $y'-y$ the \emph{reaction vector}.
Hence a reaction network of $d$ species can be represented as a simple directed graph $(\cC,\cR)$ without \emph{isolated} vertices (vertices of zero in-degree and zero out-degree), called the \emph{reaction graph}, which can be embedded in $\mathbb{N}^d_0$.  Every complex is identified with a vector of its stoichiometric coefficients in $\mathbb{N}^d_0$, e.g., $S_1+2S_2$ is identified with $e_1+2e_2$, and thus  both notations will be interchangeably used to denote a complex throughout. A reaction network is \emph{weakly reversible} if its reaction graph is a union of pairwise disjoint strongly connected components.
Let $\cC_+$ be the set of reactants, which usually determines the \emph{kinetics} of a reaction network of types of interest, including the \emph{mass-action kinetics} defined below. Throughout this paper, we assume by convention that \emph{every species appears in at least one complex}.  Hence the three sets constituting a reaction network are all known from the set of reactions. For this reason, hereafter we will simply use $\cR$ to represent a reaction network consisting of reactions in $\cR$. The linear span of the set of reaction vectors is called the \emph{stoichiometric subspace} of the reaction network. The dimension of the stoichiometric subspace of a reaction network is called the \emph{dimension of the reaction network}. A positive vector $v$ is a \emph{sub conservation law} of a reaction network $\cR$ if 
\[v\cdot(y'-y) \le 0,\quad \forall\ y\ce{->} y'\in\cR;\]
in particular, a sub-conservation law is a \emph{conservation law} \cite{MFS22} if
\[v\cdot(y'-y) = 0,\quad \forall\ y\ce{->} y'\in\cR\]
A reaction network having a (sub-)conservation law is called \emph{(sub-)conservative}. A species $S_i$ \emph{catalyzes} a reaction $y\ce{->}y'$ if $y_i=y_i'>0$.

A \emph{kinetics} of a reaction network $\cR$ is a  collection of \emph{propensity functions} $\Lambda=\{\lambda_{y\to y'}\colon \mathcal{D}\to\mathbb{R}_+\colon y\to y'\in\cR\}$,
  each of which corresponds to one reaction in the network, and the domain $\mathcal{D} = \mathbb{R}^d_+$ or $\mathcal{D} = \mathbb{N}^d_0$. A reaction $y\ce{->}y'\in\cR$ is \emph{active} in a state $x\in\mathcal{D}$ if $\lambda_{y\to y'}(x)>0$, and is \emph{inactive}  otherwise. Furthermore, a reaction is \emph{active}/\emph{inactive} in a set $\mathcal{A}\subseteq \mathcal{D}$ if it is active/inactive in every state in $\mathcal{A}$. If $\mathcal{D} = \mathbb{N}^d_0$, then $\Lambda$ is a \emph{stochastic}  kinetics with each propensity function characterizing how \emph{likely} a reaction may fire; if $\mathcal{D} = \mathbb{R}^d_+$, then $\Lambda$ is a \emph{deterministic}  kinetics with each propensity function characterizing how \emph{fast} a reaction fires. A reaction network $\mathcal{R}$ endowed with a stochastic kinetics, denoted as a pair $(\cR,\Lambda)$, is called a \emph{stochastic reaction  system} (SRS). A reaction network $\mathcal{R}$ endowed with a deterministic kinetics, denoted as a pair $(\cR,\Lambda)$, is called a \emph{deterministic reaction  system}. 
Whether $(\cR,\Lambda)$ refers to an SRS or a deterministic reaction  system will be clear from the context. For a majority of kinetics of interest in CRNT,  the propensity function of each reaction in a reaction network only depends on the reactant of the reaction. One particular  kinetics is the \emph{mass-action kinetics} defined as follows. 
 
Given the reaction network $\cR$ of stochastic mass-action kinetics, 
 for every $y\to y'\in\cR$, its propensity function is given by
 \begin{equation}\label{eq:SMAK}
 \lambda_{y\to y'}(x)=\kappa_{y\to y'}x^{\underline{y}},\quad x\in\mathbb{N}^d_0,
 \end{equation}
where $\kappa_{y\to y'}$ is the \emph{reaction rate constant} of reaction $y \ce{->} y'$.  A reaction system of stochastic mass-action kinetics is called a \emph{stochastic mass-action system}. 

Analogously, given the reaction network $\cR$ of deterministic mass-action kinetics,
 for every $y\to y'\in\cR$, its propensity function is given by
 \begin{equation}\label{eq:DMAK}
 \lambda_{y\to y'}(x)=\kappa_{y\to y'}x^{y},\quad x\in\mathbb{R}^d_+,
 \end{equation}
where $\kappa_{y\to y'}$ is again the reaction rate constant. A  reaction system of deterministic mass-action kinetics is called a \emph{deterministic mass-action system}. 

Saving the symbol $\Lambda$ of kinetics, a   (stochastic or deterministic) mass-action system can be identified with a
 \emph{weighted} directed graph $\cR$ embedded in $\mathbb{N}^d_0$, where the positive weight assigned to each edge in the graph is the reaction rate constant associated with that reaction in the reaction network. Hereafter, we do not specify ``stochastic'' or ``deterministic'' whenever {either} no kinetics applies or there is no ambiguity arising from the context.

For every reaction $y\to y'\in\cR$, the sum of stoichiometric coefficients of the reactant (i.e., $\ell_1$-norm of the reactant $\|y\|_1\coloneqq\sum_{j=1}^dy_j$) is called the \emph{molecularity} of the reaction. If the molecularity of every reaction is no more than one/two/three, then the reaction network is called \emph{unimolecular}/\emph{bimolecular}/\emph{termolecular}, respectively. Particularly, a bimolecular reaction network is called \emph{binary} if each \emph{product} also consists of at most two total molecule counts of species. Under mass-action kinetics, the sum of exponents in the propensity function of a reaction is referred to as the \emph{order} of the reaction, and hence it coincides with the molecularity of that reaction. The maximal order of all reactions of a mass-action system $\cR$ is called the \emph{order} of $\cR$. For instance, second-order mass-action systems are bimolecular,  and \emph{termolecular} mass-action systems are of third order. Although there is no distinction between order and molecularity of a reaction for a mass-action system, in general they may \emph{not} coincide, since \emph{molecularity applies to the molecular mechanism of an elementary step of a reaction while in contrast, order refers to the empirical rate law of a reaction} \cite{M67}.

\section{Regularity of reaction systems}
\subsection*{Stochastic reaction systems}
A stochastic reaction system $(\cR,\Lambda)$ can be modeled as a stochastic differential equation \cite{AK11}:
\begin{equation}\label{eq:SDE-representation}
X_t=X_0+\sum_{y\to y'\in\mathcal{R}}Y_{y\to y'}\Bigl(\int_0^t\lambda_{y\to y'}(X_s)ds\Bigr)
\end{equation}
where $X_t$ is the vector of molecule counts of the $d$ species at time $t$, $Y_{y\to y'}$ are independent unit Poisson processes associated with each reaction $y\ce{->}y'$ in the SRS, $X_0\in\mathbb{N}^d_0$ is the deterministic initial vector of species counts, and $\lambda_{y\to y'}$ is the propensity function. Note that the \emph{independence} of these Poisson processes is a consequence of the assumption that \emph{reactions cannot fire simultaneously} (with a positive probability). Hence $X_t$ is a CTMC on a state space as a subset of $\mathbb{N}^d_0$ \cite{EK09} with the extended generator of $(\cR,\Lambda)$
\[\mathcal{L}_{\cR} f(x) = \sum_{y\to y'\in\mathcal{R}}\lambda_{y\to y'}(x) (f(x+y'-y) - f(x)),\]
for every real-valued function $f$ defined on $\mathbb{N}^d_0$.

{Let $X_t$ be the {unique}\footnote{The  unique existence of a \emph{local} solution to \eqref{eq:SDE-representation} is ensured provided $\sum_{y\to y'\in\cR}\lambda_{y\to y'}(x) < \infty$ for all $x\in \mathbb{N}^d_0$ \cite[Section~2.2]{AK11}, which always holds since there are finitely many reactions in $\cR$ and propensity functions are real-valued.} solution to \eqref{eq:SDE-representation} subject to the initial state $X_0\in\mathbb{N}^d_0$. We say $X_t$ \emph{explodes} at $X_0$ if $X_t$ jumps infinitely many times within a finite time interval with a positive probability. An SRS $(\cR,\Lambda)$ is \emph{explosive} if $X_t$ explodes at some $X_0\in\mathbb{N}^d_0$, and is \emph{regular} (or \emph{non-explosive}) otherwise. {Regularity of a {stochastic reaction system} can be rephrased as global existence of solutions to \eqref{eq:SDE-representation} \cite{AK11}, a property that can be determined by the $Q$-matrix of the Kolmogorov forward equation (or so-called \emph{chemical master equation}) for} the underlying CTMC $X_t$ \cite[Corollary~2.7.3]{N98}.  Hence the set $\Lambda$ of propensity functions uniquely determines the regularity of $(\cR,\Lambda)$. 

Similarly, an SRS is \emph{transient}/\emph{null recurrent}/\emph{positive recurrent} in a communicating class $\Gamma$ of $\mathbb{N}^d_0$ if every state in $\Gamma$ is transient/null recurrent/positive recurrent for the underlying CTMC, respectively. Note that these three properties are \emph{class} properties of CTMCs \cite[Theorem~1.5.4, Theorem~3.4.3, Theorem~3.5.3]{N98}. Furthermore, we say an SRS is \emph{transient}/{\emph{null recurrent}}/\emph{positive recurrent} if it is  transient/{null recurrent}/positive recurrent in every communicating class, respectively.

\subsection*{Deterministic reaction systems}
A deterministic reaction system $(\cR,\Lambda)$ can be modeled in terms of an (autonomous) ODE \cite{AK11}, which is the \emph{mean field limit} of a net of volume-scaled SRS as the system volume tends to infinity \cite{K70,K71}:
\begin{equation}\label{eq:ODE}
\dot{x}(t) = \sum_{y\to y'\in\mathcal{R}}\lambda_{y\to y'}(x(t))(y'-y),
\end{equation}
where the unique\footnote{The unique existence of a local solution to \eqref{eq:ODE} is ascertained by local Lipschitz continuity of propensity functions, which is assumed throughout; in particular, it holds if $\cR$ is a deterministic mass-action system.} solution $x(t)$ to \eqref{eq:ODE} subject to the initial condition $x_0\in\mathbb{R}^d_+$ is the vector of concentrations of species at time $t$. We say the solution $x(t)$ \emph{blows up in finite time} at $x_0$ if there exists $0<T<\infty$ such that $\lim_{t\uparrow T}\|x(t)\|_1 = \infty$; otherwise, we say $x(t)$ \emph{globally exists} for all $t>0$. We say \eqref{eq:ODE} is \emph{regular} if its solution globally exists for all initial conditions.

\begin{definition}\label{def:regularity-of-mass-action-systems}
Given a reaction network $\mathcal{R}$, we say $\mathcal{R}$ is  a \emph{regular mass-action system} if it is regular both in the stochastic sense (when it is endowed with stochastic mass-action kinetics) and in the deterministic sense (when it is endowed with deterministic mass-action kinetics). We say $\mathcal{R}$ is  an \emph{irregular mass-action  system} if it is neither regular in the (aforementioned) stochastic sense nor regular in the (aforementioned) deterministic sense.
\end{definition}

For a given mass-action system, it is natural to ask if the regularity of the ODE model and that of the CTMC model always agree. The examples below show the two types of regularity  in general may \emph{not} be consistent.

\begin{example}\label{ex:10th-order-explosive-strongly-endotactic-RN}
Consider the following two-species reaction network \cite[Example~3.2]{ACKN20}:
\[0\ce{->[$\kappa_1$]} 2S_1\ce{->[$\kappa_2$]} 4S_1+S_2\ce{->[$\kappa_3$]} 6S_1+4S_2 \ce{->[$\kappa_4$]} 3S_1\]
Endowed with deterministic mass-action kinetics, this reaction system is \emph{permanent} (in particular, with bounded trajectories)  regardless of the reaction rate constants, and hence is regular \cite[Theorem~3.1]{CNP13}.  In contrast, the corresponding stochastic mass-action reaction system is explosive \cite{ACKN20}.
\end{example}
The above example does showcase that for given a mass-action  system, the stochastic model may  explode with a positive probability while the deterministic model may be regular. Hence it is neither regular nor irregular in the sense of Definition~\ref{def:regularity-of-mass-action-systems}.

Conversely, a deterministic mass-action system may blow up in finite time in spite of a regular stochastic model.
\begin{example}\label{ex:stochastic-regular-deterministic-irregular}
Consider the following one-species reaction network:
$$\cR\colon 2S_1\ce{<=>[$\kappa_1$][$\kappa_1$]}3S_1\ce{->[$\kappa_1$]}4S_1$$
By \cite[Theorem~4.4(i)]{WX20}, the associated stochastic mass-action system is regular regardless of the  positive rate constant $\kappa_1$. Nevertheless, the ODE associated with the deterministic mass-action system is given by \[\dot{x} = \kappa_1 x^2,\] with an explicit solution $x(t) = (x_0^{-1} - \kappa_1 t)^{-1}$ subject to the initial condition $x(0) = x_0$. Hence the ODE blows up in finite time $T = (\kappa_1 x_0)^{-1}$. 
\end{example}

Before presenting the main results on the regularity of reaction systems, it is noteworthy that second-order mass-action systems are \emph{not} always regular. 
\begin{example}\label{ex:explosive-1-species}
As a variant of Example~\ref{ex:stochastic-regular-deterministic-irregular}, consider the following second-order reaction network
\[\mathcal{R} \colon 2S_1\ce{->[$\kappa_1$]}3S_1\]
As a deterministic mass-action system, $\mathcal{R}$ is dynamically equivalent to the system (i.e., with the same associated ODE) in Example~\ref{ex:stochastic-regular-deterministic-irregular} (assuming the two $\kappa_1$ in respective reaction networks are identical), and hence blows up in finite time. In contrast to the stochastic mass-action system in Example~\ref{ex:stochastic-regular-deterministic-irregular}, the stochastic mass-action system in this example is explosive \cite[Theorem~4.4(i)]{WX20}, which can be used to model a pure birth process with a quadratic growth rate.
\end{example}

Even in the light of Example~\ref{ex:explosive-1-species}, one might still optimistically anticipate that all CTMCs proposed in the literature modeling world-real phenomena are regular. Indeed, this is  \emph{not} true. For instance, the so-called ``runaway model'' \cite{BA16}\---a Verhulst logistic model with
bursty reproduction (Example~\ref{ex:run-away}),  
turns out to be \emph{explosive} \cite{XHW23}.

Before proceeding to the main results, let us recall a well-known drift criterion for the regularity of CTMCs.
\begin{proposition}\label{prop:regularity}\cite[Theorem~2.1]{MT93}
Let $X_t$ be a CTMC on an infinite state space $\mathcal{X}\subseteq\mathbb{Z}^d$. Assume there exist constants $C_1,C_2\ge0$, and a non-negative norm-like\footnote{$V\colon \mathbb{R}^d\to \mathbb{R}_+$ is a \emph{norm-like} function if $\lim_{\|x\|_1\to\infty}V(x)=+\infty$.} function $V\colon \mathcal{X}\to\mathbb{R}$ such that  
\[\mathcal{L}V(x)\le C_1V(x)+C_2,\quad \forall x\in\mathcal{X},\]
{where $\mathcal{L}$ is the extended generator of $X_t$.} Then $X_t$ is regular on $\mathcal{X}$. Such a function $V$ is called a \emph{Lyapunov function for the regularity of $X_t$}. 
\end{proposition}

Next, we identify a class of regular SRS in terms of a simple condition.

For every {positive vector $u\in\mathbb{R}^d$}, let
$$\cR_{u,+}\coloneqq\{y\ce{->} y'\in\cR \colon y'>_uy\}$$ be the (possibly empty) subset of \emph{reactions} as bona fide vectors \emph{increasing along the $u$-direction}. 

\begin{theorem}\label{thm:regularity}
Let $(\cR,\Lambda)$ be an SRS with $\Lambda=\{\lambda_{y\to y'}\colon y\ce{->}y'\in\cR\}$. 
If there exists a positive vector $u\in\mathbb{R}^d$  
such that
\begin{equation}\label{eq:upper-bound}
\sup_{y\to y'\in\cR_{u,+}}\sup_{x\in\mathbb{N}^d\setminus\{0\}}\frac{\lambda_{y\to y'}(x)}{\|x\|_1} < \infty,
\end{equation}
then $(\cR,\Lambda)$ is regular. In particular, $(\cR,\Lambda)$ is regular if $\cR_{u,+}=\emptyset$ for some positive vector $u$.
\end{theorem}
Roughly speaking, an SRS is regular if there exists a positive direction such that all reactions that increase along this direction have non-strictly sub-linear propensity functions.
\begin{proof}
Let $V(x)=x^T\cdot u$. It follows from \eqref{eq:upper-bound} that
\begin{align*}
\mathcal{L}_{\mathcal{R}}V(x)=&\sum_{y\to y'\in\cR}\lambda_{y\to y'}(x) \Bigl((x+y'-y)^T\cdot u)-x^T\cdot u\Bigr)\\
\le&\sum_{y\to y'\in\cR}\lambda_{y\to y'}(x)(y'-y)^T\cdot u\\
\le&\sum_{y\to y'\in\cR_{u,+}}\lambda_{y\to y'}(x)(y'-y)^T\cdot u\le \widetilde{C}V(x)
\end{align*}
for some constant $\widetilde{C}\in\mathbb{R}$. Hence,  by Proposition~\ref{prop:regularity}, $(\cR,\Lambda)$ is regular.

In particular, if $\cR_{u,+}=\emptyset$, then $$\sup_{y\to y'\in\cR_{u,+}}\sup_{x\in\mathbb{N}^d\setminus\{0\}}\frac{\lambda_{y\to y'}(x)}{\|x\|_1}=0,$$ which immediately yields that $(\cR,\Lambda)$ is  regular.
\end{proof}

\begin{theorem}
\label{thm:regularity-ODE}
\!\!\! Let $(\cR,\Lambda)$ be a deterministic reaction system with  $\Lambda=\{\lambda_{y\to y'}\colon$ $y\ce{->}y'\in\cR\}$.
If there exists a positive vector $u\in\mathbb{R}^d$
 such that
\begin{equation}\label{eq:upper-bound-ODE}
\sup_{y\to y'\in\cR_{u,+}}\sup_{x\in\mathbb{R}^d_+\setminus\{0\}}\frac{\lambda_{y\to y'}(x)}{\|x\|_1} < \infty,
\end{equation}
then $(\cR,\Lambda)$ is regular. In particular, $(\cR,\Lambda)$ is regular if $\cR_{u,+}=\emptyset$ for some positive vector $u$.
\end{theorem}

\begin{proof}
Let $V(x)=x^T\cdot u$. It follows from \eqref{eq:upper-bound-ODE} that
\begin{align*}
\frac{\dot{V}(x(t))}{dt}=&\sum_{y\to y'\in\cR}\lambda_{y\to y'}(x(t)) \Bigl((x(t)+y'-y)^T\cdot u)-x(t)^T\cdot u\Bigr)\\
=&\sum_{y\to y'\in\cR}\lambda_{y\to y'}(x(t))(y'-y)^T\cdot u\\
\le&\sum_{y\to y'\in\cR_{u,+}}\lambda_{y\to y'}(x(t))(y'-y)^T\cdot u\le \widetilde{C}V(x(t))
\end{align*}
for some constant $\widetilde{C}\in\mathbb{R}$, which further implies by Gronwall inequality that $$x(t) \le x_0 e^{\widetilde{C} t},$$ where $x_0$ is the initial concentration vector. Hence $(\cR,\Lambda)$ is regular.

With a similar argument as in the proof of Theorem~\ref{thm:regularity}, one can show that $(\cR,\Lambda)$ is regular if $\cR_{u,+}=\emptyset$.
\end{proof}

\begin{theorem}\label{thm:1st-order-regularity}
For a given reaction system $(\cR,\Lambda)$, let $\cR_*$ be the (possibly empty) sub reaction network of zero-th- and first-order reactions. Assume $\Lambda|_{\cR_*}$ is of mass-action\footnote{It is stochastic mass-action if $\Lambda$ is a stochastic kinetics, and it is deterministic mass-action if $\Lambda$ is a deterministic kinetics.} and $$\lambda_{y\to y'}(x) = \kappa_{y\to y'}\eta_y(x),\ \forall y\ce{->}y'\in\cR\setminus\cR_*$$
for a family of non-negative functions $\{\eta_y\colon y\in\hat{\cC}_+\}$, where $\hat{\cC}_+$ is the set of reactants of $\cR\setminus\cR_*$, and $\kappa_{y\to y'}$ are reaction rate constants. 
Then the $(\cR,\Lambda)$ is regular if one of the following three conditions holds:
\begin{itemize}
    \item[(C1)] all reactions where the total molecule counts increase is a subset of $\cR_*$;
     \item[(C2)] $\cR\setminus\cR_*$ is a (possibly empty) sub-conservative reaction network; and
     \item[(C3)] $\cR\setminus\cR_*$ is non-empty, and there exists a species $S_i$ in $\cR\setminus\cR_*$ such that the molecule counts of $S_i$ are non-increasing in every reaction in $\cR\setminus\cR_*$, and if there is a reaction $y\ce{->}y'$ in $\cR\setminus\cR_*$ that $S_i$ catalyzes  then there exists another reaction in $\cR\setminus\cR_*$ with the same reactant $y$ in which the total molecule counts of $S_i$ decreases.
\end{itemize}

In particular, every first-order mass-action system is regular.
\end{theorem}
\begin{proof}
A first-order mass-action system fulfills condition (C1) and hence its regularity follows from case (C1). Below we prove cases (C1)-(C3). 

Case (C1) follows from case (C2) with sub-conservation law $\mathbf{1}$. For case (C2), let $u$ be a sub-conservation law of $\cR\setminus\cR_*$. Then $\cR_{u,+}\subseteq\cR_*$, and  Theorem~\ref{thm:regularity} and Theorem~\ref{thm:regularity-ODE} apply to $\cR$ with vector $u$. 

Case (C3) does not follow directly from Theorem~\ref{thm:regularity} and Theorem~\ref{thm:regularity-ODE} while shares similar proofs with a linear Lyapunov function. Let $S_i$ be such a species as specified in (ii). As we will see below, one can choose $u=(1,\ldots,1,u_i,1,\ldots,1)>0$ with some $u_i>1$ to be determined  such that $V(x)=u\cdot x$ is a Lyapunov function for regularity of $\cR$ in both senses. 
Indeed, let $\cR_1\subseteq\cR\setminus\cR_*$ be the possibly empty set of reactions where the molecule count of $S_i$ remains. Hence $\cR_2=(\cR\setminus\cR_*)\setminus\cR_1$ consists of reactions where the molecule count of $S_i$ decreases. Let $\cC_{i,+}$ be the set of reactants of $\cR_i$ for $i=1,2$. By condition (C2), $\cC_{1,+}\subseteq\cC_{2,+}$. 
Let $$\cR_3=\{y\ce{->} y''\in\cR_2\colon y\in\cC_{1,+}\}$$ 

Choose $u_i$ large enough  satisfying the following constraints:  
\begin{align*}
   \sum_{y\to y'\in\cR_1} \kappa_{y\to y'}\sum_{j\neq i}(y'_j-y_j)\le \sum_{y\to y''\in\cR_3}\kappa_{y\to y''}\bigl(u_i(y_i-y''_i)-\sum_{j\neq i}(y''_j-y_j)\bigr), \forall y\in\cC_{1,+};
\end{align*} 
\begin{align*}
   u_i(y_i'-y_i)+\sum_{j\neq i}(y'_j-y_j)\le0,\quad y\ce{->}y'\in\cR_2\setminus\cR_3 
\end{align*}
Then based on the above constraints, we have
\begin{align*}
    \mathcal{L}_{\cR} V(x) = & \sum_{y\to y'\in\cR_*}\kappa_{y\to y'}x^{\underline{y}}u\cdot(y'-y)\\
    & +\sum_{y\in\cC_{1,+}}\eta_y(x)\Bigl(\sum_{y\to y'\in\cR_1}\kappa_{y\to y'}u\cdot (y'-y)+\sum_{y\to y''\in\cR_3}\kappa_{y\to y''}u\cdot (y''-y)\Bigr)\\
    &+ \sum_{y\to y'\in\cR_2\setminus\cR_3}\kappa_{y\to y'}\eta_y(x)u\cdot (y'-y)\\
    \le& \sum_{y\to y'\in\cR_*}\kappa_{y\to y'}x^{\underline{y}}u\cdot(y'-y)\le CV(x)
\end{align*}
for some constant $C>0$. Hence $V$ is a Lyapunov function for non-explosivity of the stochastic model. One can show analogously that it is also a Lyapunov function for regularity of the ODE model. 
\end{proof}

\begin{remark}\label{re:note-on-kinetics}
    The assumption on the kinetics of $\cR\setminus\cR_*$ is only used in Theorem~\ref{thm:1st-order-regularity}(C3).
\end{remark}

Below we have the following alternative interpretation of Theorem~\ref{thm:1st-order-regularity}(C1) in terms of an auxiliary reaction network. For the ease of exposition, we confine the statement to mass-action kinetics.

Given a reaction network $\cR$ \emph{without conservation law $\mathbf{1}$}, remove all reactions $y\to y'\in\cR$ conserving the total molecules (i.e., $\|y\|_1=\|y'\|_1$), one obtains the following \emph{non-empty} auxiliary one-species reaction network $$\cR_{\diamond}\coloneqq\{\|y\|_1S\ce{->[\underset{z\to z'\in\cR\colon \|z\|_1=\|y\|_1,\ \|z'\|_1=\|y'\|_1}{\sum}\kappa_{z\to z'}]} \|y'\|_1S\colon y\to y'\in\cR,\ \|y\|_1\neq\|y'\|_1\}$$  

\begin{corollary}\label{cor:one-species-interpretation}
Let $\cR$ be a reaction network. Assume $\cR$ has no conservation law $\mathbf{1}$. Let $\cR_\diamond$ be its auxiliary one-species reaction network. For any reaction $mS\to nS\in\cR_\diamond$, if $m<n$ implies $m\le1$, then $\cR$ as a mass-action system is regular.
\end{corollary}
\begin{proof}
    Since for any reaction $mS\to nS\in\cR_\diamond$, if $m<n$ implies $m\le1$, we know every reaction in $\cR$ that increases the total molecule counts is of zero-th- or first-order, and hence is contained in $\cR_*$. Then the conclusion follows from Theorem~\ref{thm:1st-order-regularity}(C1).
\end{proof}

\begin{example}
    Consider the following reaction network $$\mathcal{R}\colon S_1+S_2\ce{->}S_3\quad S_2\ce{<=>}S_3$$
    Removing the pair of reversible reactions $S_2\ce{<=>}S_3$ that conserve the total molecules, and replacing $S_1+S_2$ by $2S$ and $S_3$ by $S$ in the only remaining reaction, we obtain the following one-species auxiliary network:
    \[\mathcal{R}_\diamond\colon 2S\ce{->}S\] 
     Since $\mathcal{R}_\diamond$ contains no reaction $mS\to nS$ with $m<n$, by Corollary~\ref{cor:one-species-interpretation}, the associated mass-action system of $\cR$ is regular.
\end{example}

Below we continue to provide more examples to demonstrate the subtle difference among the three cases in Theorem~\ref{thm:1st-order-regularity}. First, consider an example where Theorem~\ref{thm:1st-order-regularity}(C1) is \emph{not} applicable; instead Theorem~\ref{thm:1st-order-regularity}(C2) and Theorem~\ref{thm:1st-order-regularity}(C3) apply.

\begin{example}
    Consider the following reaction network:
    \[\cR\colon S_1+S_2\ce{->}3S_3\quad S_2\ce{<=>}S_3\]
    Since the first reaction is bimolecular, and the total molecule counts increase in this reaction, \ref{thm:1st-order-regularity}(i) fails. However, $\cR$ is conservative with a conservation law $(2,1,1)$ (and hence so is $\cR_*$ with the same conservation law). Hence the associated mass-action system is regular by Theorem~\ref{thm:1st-order-regularity}(C2). Indeed, species $S_1$ (or species $S_2$) decreases in the unique reaction of higher than first order, and hence the regularity also follows from Theorem~\ref{thm:1st-order-regularity}(C3).
\end{example}

Next, we give an example where Theorem~\ref{thm:1st-order-regularity}(C1) and Theorem~\ref{thm:1st-order-regularity}(C3) fail while Theorem~\ref{thm:1st-order-regularity}(C2) applies.

\begin{example}
    Consider the following reaction network:
    \[\cR\colon S_1+S_2\ce{->}2S_1+S_3\quad 2S_1\ce{->}S_2+2S_3\]
    It is easy to observe that the total molecule counts increase in both reactions; moreover, the molecule count of $S_1$ increases in the first reaction while decreases in the second reaction, and conversely, the molecule count of $S_2$ decreases in the first reaction while increases in the second reaction; the molecule count of $S_3$ increases in both reactions. Hence neither (C1) nor (C3) in Theorem~\ref{thm:1st-order-regularity} is fulfilled. Nevertheless, it is easy to verify that $\cR\setminus\cR_*$ has a sub-conservation law $u=(4,6,1)$. By Theorem~\ref{thm:regularity}(C2) and Remark~\ref{re:note-on-kinetics}, $\cR$ is regular regardless of its kinetics.
\end{example}

For an example where neither Theorem~\ref{thm:1st-order-regularity}(C1) nor Theorem~\ref{thm:1st-order-regularity}(C2) applies while Theorem~\ref{thm:1st-order-regularity}(C3) does apply, see the immunological model in Section~\ref{sect:appl}.

\section{{Regularity of second-order mass-action systems}}\label{sect:Endotactic_system}

Motivated by the concern on the properties of binary reaction networks and endotactic reaction networks aforementioned in the Introduction, we will apply the main results derived in the previous section to yield the regularity of these reaction networks.

\subsection*{Endotacticity}
First, let us recall the definitions of endotactic and strongly endotactic reaction networks. Both types of reaction networks were originally proposed \cite{CNP13,GMS14} to address the dissipativity of the associated deterministic reaction systems.

\begin{definition}
Given $u\in\mathbb{R}^d$, for any $y,z\in\mathbb{R}^d$, we write $y>_uz$ or $z<_uy$ if $(y-z)^T\cdot u>0$. Similarly, we define $y=_uz$, $y\ge_uz$, and $z\le_uy$.
In this way, $u$ induces a total order on $\mathbb{R}^d$.  Two vectors $y,z\in\R$ are \emph{$u$-equal} if $y=_uz$. 
Given any subset $A\subseteq\R$, an element $y\in A$ is \emph{$u$-maximal} in $A$ if $y\ge_uz$ for all $z\in A$.
\end{definition}

\begin{definition}
Given a reaction network $\cR$ and a vector $u\in\mathbb{R}^d$, let $\cR_u$ be the (possibly empty) set of reactions
whose reaction vectors are \emph{not} orthogonal to $u$ and $\cC_{u,+}$ be the (possibly empty) set of reactants 
in $\cR_u$.
 Any reaction $y\ce{->[]} y'\in\cR$ with a $u$-maximal reactant in $\cC_{u,+}$ satisfying $y<_uy'$ is called a \emph{$u$-endotacticity-violating reaction} of $\cR$, or simply called an \emph{endotacticity-violating reaction} of $\cR$. We say $\cR$ is \emph{$u$-endotactic} if it has no $u$-endotacticity-violating reaction.  A $u$-endotactic reaction network $\cR$ is \emph{strongly $u$-endotactic} if there exists a reaction $y\to y'\in\cR$ such that $y>_uy'$ and $y$ is $u$-maximal in $\cC_+$.
\end{definition}
\begin{definition}
$\cR$ is \emph{lower-endotactic} (\emph{strongly lower-endotactic}, respectively) if $\cR$ is $u$-endotactic (strongly $u$-endotactic, respectively) for every $u\in\mathbb{R}^d_+$.
\end{definition}
\begin{definition}
$\cR$ is \emph{endotactic} (\emph{strongly endotactic}, respectively) if $\cR$ is $u$-endotactic (strongly $u$-endotactic, respectively) for every $u\in\mathbb{R}^d$. In other words, \emph{$\cR$ is endotactic if and only if $\cR$ has no endotacticity-violating reaction}.
\end{definition}
By definition, it is apparent that endotacticity of a reaction network is a property of the reaction graph disrespecting the edge weights.

Only based on the definition, it may turn out non-trivial to justify that a reaction network is endotactic, since literally one needs to show $u$-endotacticity for infinitely many vectors $u$. However,  for either 2-dimensional reaction networks \cite[Proposition~4.1]{CNP13} or first-order reaction networks \cite[Theorem~5.2]{X26}, to conclude the endotacticity of a reaction network, it suffices to check $u$-endotacticity for only \emph{finitely} many vectors $u$. Reaction networks satisfying other more readily verifiable conditions purely on the reaction graph are known to be endotactic \cite[Theorem~4.4]{X26}; for instance, \emph{every  weakly reversible reaction network is endotactic} \cite[Lemma~4.5]{CNP13}.

Since the main results of this paper impose conditions on  $u$-endotacticity for a particular vector $u$, to save some extra notations, we will omit stating the known  criteria for the endotacticity of reaction networks,  particularly for 2-dimensional or first-order reaction networks. We refer the interested reader to \cite{CNP13,GMS14,X26} for more discussions.
\begin{example}\label{ex:illustration-of-u-endotacticity}
 Consider the following Brusselator reaction network \cite{F87}:
\[
\cR\colon 0\ce{<=>}S_1\ce{->} S_2\quad 2S_1+S_2\ce{->} 3S_1
\]
\begin{enumerate}
\item[(i)] Let $u=(1,1)$. Then $\cR_u=\{0\ce{<=>}S_1\}$,  and hence $\cC_{u,+}=\{0,e_1\}$. Since $e_1>_u0$, there exists no $u$-endotacticity-violating reaction. Hence $\cR$ is $u$-endotactic.
\item[(ii)] Let $v=(2,1)$. Then $\cR_{v}=\cR$ since no reaction vector is orthogonal to $v$. Moreover, $\cC_{v,+}=\{0,e_1,2e_1+e_2\}$, and $2e_1+e_2$ is the unique $v$-maximal reactant in $\cC_{v,+}$. Observe that  $2S_1+S_2\ce{->} 3S_1$ is a $v$-endotacticity-violating reaction since $2e_1+e_2<_v3e_1$, which yields that $\cR$ is {not $v$-endotactic and hence not endotactic}. In addition, it is straightforward to show that $$\cR_{v,+}=\{0\ce{->}S_1\quad 2S_1+S_2\ce{->} 3S_1\}$$
\end{enumerate}
\end{example}

Below we present regularity of a wide class of second-order mass-action systems.

\begin{theorem}\label{thm:endotactic-regularity}
\begin{itemize}
    \item[(a)] Every binary  mass-action system is regular; and
    \item[(b)] every  second-order $\mathbf{1}$-endotactic mass-action system is regular.
\end{itemize}
 In particular,
\begin{enumerate}
\item[(i)] every second-order lower-endotactic mass-action system is regular;
\item[(ii)] every second-order endotactic mass-action system is regular;
\item[(iii)] every second-order strongly lower-endotactic mass-action system is regular;
\item[(iv)] every second-order strongly endotactic mass-action system is regular; and
\item[(v)] every second-order weakly reversible mass-action system is regular.
\end{enumerate}
\end{theorem}
\begin{proof}
We first prove (a) and (b) by showing that for every second order reaction $y\ce{->} y'$ in such networks, we have
\begin{equation}\label{eq:total-no-increase}
(y'-y)^T\cdot\mathbf{1}\le0,
\end{equation}i.e., the total molecule counts do not increase. Then the conclusion follows from Theorem~\ref{thm:1st-order-regularity}(C1).
\begin{itemize}
\item[(a)]  \eqref{eq:total-no-increase} follows from \[\|y'\|_1\le 2=\|y\|_1,\] since every product has no more than two molecules;
    \item[(b)] \eqref{eq:total-no-increase} follows from the $\mathbf{1}$-endotacticity.
\end{itemize}
The remaining conclusions in (i)-(v) follow from (b) based on the fact that these networks are all $\mathbf{1}$-endotactic \cite{CNP13,GMS14}.
\end{proof}

\begin{remark}\label{re:comparison-blow-up-versus-explosivity-for-complex-balanced-systems}
Another class of regular mass-action systems are \emph{complex-balanced} mass-action systems. Complex-balanced mass-action systems are regular in both the  stochastic sense \cite{ACKK18} and the  deterministic sense \cite{F87} (see also \cite{CNP13}).
\end{remark}
\begin{remark}
Example~\ref{ex:10th-order-explosive-strongly-endotactic-RN}  
is an explosive (strongly) endotactic stochastic mass-action reaction system of order 10.
 It remains open \emph{what is the minimal order for an endotactic stochastic mass-action system to be explosive}.
\end{remark}
\begin{remark}
It is known that strongly endotactic deterministic mass-action systems \cite{CNP13}, as well as 2-species endotactic deterministic mass-action systems \cite{GMS14} and first order endotactic deterministic mass-action systems \cite{X26}, are permanent and hence are regular. In addition, \emph{asiphonic strongly endotactic} deterministic mass-action systems are also regular \cite[Theorem~4.11 and Remark~3.8] {ADE18b}; it is noteworthy that ``strong endotacticity'' therein is weaker than the concept in this paper.
\end{remark}

Although endotactic reaction networks are a large class of reaction networks which are regular, $\mathbf{1}$-endotacticity is  \emph{unnecessary} for the regularity of second-order mass-action systems.

\begin{example}
    Consider the following second-order one-species mass-action system:
    \[\mathcal{R}\colon S\ce{<=>[$\kappa_0$][$\kappa_1$]}2S\ce{->[$\kappa_2$]}3S\]
    $\mathcal{R}$ is not $1$-endotactic, while it is non-explosive for $\kappa_2 < \kappa_1$ when modeled stochastically by \cite[Theorem~1]{XHW23}; moreover, it is easy to verify that it is also regular under this condition when modeled deterministically. Hence $\mathcal{R}$ is a regular mass-action system for $\kappa_2 < \kappa_1$.
\end{example}

For higher (than second) order mass-action systems, even \emph{strong} $\mathbf{1}$-endotacticity becomes insufficient for the regularity.
\begin{example}
Consider the following termolecular reaction network
\[\cR\colon 0\ce{<=>[$\kappa_0$][$\kappa_{-1}$]} S_1\quad 2S_1\ce{->[$\kappa_1$]}3S_1\quad 3S_2\ce{->[$\kappa_2$]}2S_2\]
It is readily verified that $\cR$ is strongly $(1,1)$-endotactic. Nevertheless, for the stochastic mass-action system $\cR$, $\Gamma_i=\mathbb{N}_0\times\{i\}$ for $i=0,1,2$ are the only three closed communicating classes in the ambient space $\mathbb{N}^2_0$, and the reaction $3S_2\ce{->[$\kappa_2$]}2S_2$ is inactive in $\Gamma=\cup_{i=0}^2\Gamma_i$, and hence for a stochastic mass-action system $\cR$, its extended generator coincides with that of its sub reaction network  $$\cR'\colon 0\ce{<=>[$\kappa_0$][$\kappa_{-1}$]} S_1\quad 2S_1\ce{->[$\kappa_1$]}3S_1$$ confined to $\Gamma$. Note that as a one-dimensional stochastic mass-action system, $\mathcal{R}'$ is explosive  \cite[Theorem~4.4(i)]{WX20}. Hence the stochastic mass-action system $\cR$ is also explosive. One can easily show by straightforward calculations that for $\cR$ as a deterministic mass-action system, the concentration of species $S_1$ blows up in finite time. In sum, $\cR$ is an irregular mass-action system.
\end{example}

\begin{theorem}\label{thm:1-endotactic-positive-recurrent-equivalent-to-existence-of-SD}
Let $\cR$ be a second-order stochastic mass-action system.  If $\cR$ is binary or $\mathbf{1}$-endotactic, then $\cR$ has at most one stationary distribution in each closed communicating class.
\end{theorem}

\begin{proof}
Assume w.l.o.g. that there exists a closed communicating class $\Gamma$, and $\cR$ has a stationary distribution in $\Gamma$. Otherwise, the conclusion holds trivially. Then $X_t$ restricted to $\Gamma$ is irreducible. It follows from Theorem~\ref{thm:endotactic-regularity} that $X_t$ is regular. By \cite[Theorem~3.5.3]{N98}, a regular irreducible CTMC with a stationary distribution is positive recurrent, which implies the uniqueness of the stationary distribution in $\Gamma$.
\end{proof}
\begin{theorem}\label{thm:reducing-positive-recurrence-to-existence-of-SD}
\ Every second-order weakly reversible stochastic mass-action system is positive recurrent if and only if there exists a stationary distribution in every communicating class.
\end{theorem}
\begin{proof}
Note that weakly reversible reaction networks of $d$ species are \emph{essential} in the sense that every communicating class in the ambient state space $\mathbb{N}^d$ is closed \cite[Theorem~4.6]{PCK14} (see also \cite{XHW22}). Then the conclusion follows from Theorem~\ref{thm:1-endotactic-positive-recurrent-equivalent-to-existence-of-SD}, again in the light of the fact that every weakly reversible reaction network is $\mathbf{1}$-endotactic \cite{CNP13}.
\end{proof}

Note that \emph{$\mathbf{1}$-endotacticity is insufficient for the existence of a closed} \emph{communicating class}.
\begin{example}
Consider the following stochastic mass-action system $$\cR \colon S_1+S_2\ce{->} 0\ce{->} S_2$$ It is readily verified that $\cR$ is $(1,1)$-endotactic  (indeed, it is lower endotactic). Moreover, the two reactions together facilitate first exhausting species $S_1$ and then increasing the molecule counts of $S_2$ indefinitely: Every non-zero state in $\mathbb{N}^2_0$ strictly below the diagonal $\{(n,n)\}_{n\in\mathbb{N}_0}$ leads to a state on the diagonal  by possibly repeatedly applying the reaction $0\ce{->} S_2$; every state either on or strictly above the diagonal leads to a state on the boundary $\{0\}\times\mathbb{N}_0$ by  possibly repeatedly applying the reaction $S_1+S_2\ce{->} 0$; and every state on the boundary $\{0\}\times\mathbb{N}_0$ leads to all states above it by the reaction $0\ce{->} S_2$. Hence every state is in an open communicating class and $\cR$ is transient \cite[Theorem~1.5.5\ and Theorem~3.4.1]{N98} in a trivial manner due to the  underlying structure of the ambient space.
\end{example}

Indeed, even with the existence of a closed communicating class, $\mathbf{1}$-endotacticity turns out to be still insufficient for the recurrence for second-order stochastic mass-action systems.

\begin{example}
Consider the stochastic mass-action system
\[\cR\colon 0\ce{<=>[$\kappa_1$][$\kappa_2$]}S_2\ce{->[$\kappa_3$]}2S_2\ce{<-[$\kappa_4$]}S_1+S_2\ce{->[$\kappa_5$]}S_2\]
It is straightforward to verify that $\cR$ is $(1,1)$-endotactic; moreover, $\Gamma=\{0\}\times\mathbb{N}_0$ is the unique closed communicating class, and $\{n\}\times\mathbb{N}_0$ is an open communicating class for every $n\in\mathbb{N}$. Hence restricted to $\Gamma$, its extended generator coincides with that of the sub reaction network of $\cR$ as a  stochastic mass-action system $$\cR'\colon 0\ce{<=>[$\kappa_1$][$\kappa_2$]} S_2 \ce{->[$\kappa_3$]} 2S_2$$   Since $\cR'$ is transient if and only if $\kappa_3>\kappa_2$ or $\kappa_1>\kappa_2=\kappa_3$ \cite[Theorem~3(i)]{XHW23}, $\cR$ is also transient when $\kappa_3>\kappa_2$ or  $\kappa_1>\kappa_2=\kappa_3$.
\end{example}
\begin{remark}
In \cite{AM18}, a 7th order transient strongly endotactic  stochastic mass-action system was elegantly constructed. The author is unaware of  \emph{what is the minimal order for an endotactic  stochastic mass-action system to be transient in a closed communicating class}.
\end{remark}

\begin{corollary}\label{cor:positive-recurrence}
Let $\cR$ be a second-order stochastic mass-action system. Assume $\cR$ is either binary or $\mathbf{1}$-endotactic. Then $\cR$ is positive recurrent in a closed communicating class $\Gamma$ if and only if there exists a stationary distribution in $\Gamma$.
\end{corollary}
\begin{proof}
The forward implication (``only if'' part) is obvious, and the backward implication (``if'' part) follows immediately from Theorem~\ref{thm:endotactic-regularity} and \cite[Theorem~3.5.3]{N98}.
\end{proof}

\begin{remark}
In \cite{AM18}, a 7th order null recurrent strongly endotactic  stochastic mass-action system was also elegantly  constructed whose reaction graph is akin to that of a positive recurrent strongly endotactic stochastic mass-action system. Note that every stochastic mass-action system has a  \emph{stationary measure} in every closed communicating class \cite{WX23}. The author is ignorant of \emph{what is the minimal order for an endotactic  stochastic mass-action system to admit no stationary distribution (i.e., no finite stationary measure)} in an infinite closed communicating class.
\end{remark}

In contrast to second-order weakly reversible stochastic mass-action systems, third-order weakly reversible stochastic mass-action systems may \emph{not} admit a linear Lyapunov function for regularity.
\begin{example}\label{ex:linear-function-fails-for-3rd-order-WR}
Consider the following termolecular two-species stochastic mass-action system 
\[\cR\colon S_1+S_2\ce{<=>[$\kappa_4$][$\kappa_3$]}2S_1+S_2\ce{<=>[$\kappa_2$][$\kappa_1$]}2S_1\] Note that $\mathbb{N}\times\mathbb{N}_0\setminus\{(1,0)\}$ is the unique non-singleton communicating class in $\mathbb{N}^2_0$, and it is closed. Since $\cR$ is weakly reversible and is of \emph{deficiency zero}\footnote{\emph{Deficiency} is the number of complexes minus the sum of the dimension of $\cR$ and the number of non-singleton strongly connected components of $\cR$. In this example, it equals $3-2-1=0$.}, it is complex-balanced \cite[Theorem~4A]{H72}, and hence is a regular mass-action system by Remark~\ref{re:comparison-blow-up-versus-explosivity-for-complex-balanced-systems}. Let $\cL_{\cR}$ be the extended generator associated with $\cR$. For every positive vector $u\in\mathbb{R}^2$, let
$V_u(x) =x^T\cdot u$.  Then for all large $x_1\in\mathbb{N}\setminus\{1\}$, \begin{align*}
\cL_{\cR} V_u(x_1,0) =& \kappa_1x_1(x_1-1)u_2
\end{align*}
Obviously, $V_u$ is not a Lyapunov function for the regularity of $\cR$ for any $u\in\mathbb{R}^2_{+}$ (note that $V_u$ is not \emph{norm-like} for any $u$ on the boundary of $\mathbb{R}^2_{+}$). Essentially, this is because the reaction $2S_1\ce{->[$\kappa_1$]}2S_1+S_2$ of a strictly \emph{super-linear} propensity function increases (the molecule counts of $S_2$) along any positive direction. Analogously, one can show that $\cR$ as a deterministic mass-action system admits no linear Lyapunov function for its global existence of solutions.
\end{example}
Below we provide an example in action {to demonstrate} how to apply Corollary~\ref{cor:positive-recurrence} to prove positive recurrence.
\begin{example}\label{ex:illustration-with-existence-of-SD}
Consider the following {binary} reaction network taken from \cite[Example~1]{HHKK23}
\[\cR\colon {S_3}\ce{<=>[$\kappa_1$][$\kappa_2$]}0\ce{->[$\kappa_3$]} S_1\ce{->[$\kappa_4$]} S_2\quad 2S_1\ce{->[$\kappa_5$]} S_1+S_2\quad 2S_2\ce{->[$\kappa_6$]} S_2+S_3\]
It is readily verified that the CTMC associated with $\cR$ is irreducible on $\mathbb{N}^3_0$; moreover, there exists a stationary distribution (see \cite{HHKK23} for the closed form of the stationary distribution). By Corollary~\ref{cor:positive-recurrence}, $\cR$ is positive recurrent.
\end{example}

\section{Applications}\label{sect:appl}

In this section, we will apply the main results in the previous two sections to obtain regularity of various models in the literature stemming from real-world applications.

\subsection*{Preparations}
\emph{Open reaction networks} in the sense of Feinberg \cite{F87} refer to those that contain ``pseudo-reactions'' (i.e., reactions where either the reactant or the product is the zero complex) to account for the case where mass exchange with the ambient is possible. Pseudo-reactions of a zero reactant account for reactions where certain species in the reactor vessel is produced in the feed stream at constant molar concentrations. Similarly, pseudo-reactions of a zero product model removal of the species in the reactant from the reactor vessel by the effluent stream. \emph{Closed reaction systems}, as  opposed to \emph{open reaction systems}, assume  conservation of mass in terms of the existence of a conservation law for the reaction network \cite{F19,MFS22}. Since closed reaction systems, when modeled stochastically, are \emph{finite} state CTMCs owing to the existence of a conservation law, whose regularity trivially holds. To avoid this triviality while {demonstrating} the applicability of the results established in the previous section, we will apply our results to various biological models that can be represented by \emph{open reaction networks}. Throughout this section,  to save the effort in notation, \emph{the same symbol $\cR$ is slightly abused to only refer to that specific reaction network within each subsection.}

Below, models with an \emph{italic subsection title} were originally proposed as  \emph{deterministic} reaction systems in the references, where reactions are assumed to occur in an isothermal homogeneous continuous
flow stirred tank reactor (CFSTR), and the biological appropriateness might have been justified therein.
It is noteworthy that in contrast, we simply take the reaction mechanism as a reaction network; in this way, we can derive a stochastic analogue, by taking into account the randomness induced by possible \emph{low} abundance of species, \emph{without justifying its biological appropriateness}. In spite of this, empirical evidence does support that certain stochasticity may exist in a few  models, e.g., the Lotka-Volterra model \cite{KK88}, the Brusselator model \cite{BFP10}, the glycolysis model \cite{JBP02}, the enzymatic model \cite{GMK17}, and metabolic pathways in general \cite{LH07}.

\emph{Unless stated otherwise, within this section, all reaction systems are mass-action systems in both the stochastic sense and the deterministic sense.}

It is plausible that one motivation of mathematical biology is to unveil the common mathematical law of similar biological models. Since it is arguably closer to the truth (than the negation of the following statement) that many phenomena in nature share the same generic principle or law, similar or even identical chemical mechanisms may exist in more than one forms in nature. For this consideration, though the assumption of low abundance which justifies the stochastic models seems contradictory to the assumption of high copy numbers of species which underlies deterministic reaction systems, the discussion about the stochastic reaction systems sharing the same chemical mechanisms with the  deterministic reaction systems in the references may not be entirely meaningless, even from the perspective of biology.

As we will see below, our main results apply
to all the models, which demonstrates its wide applicability.  

\subsection*{SIRS epidemic model}
Consider the following susceptible-infected-recovered (SIRS) epidemic model \cite{GBK14} in terms of the following reaction network:
\begin{align*}
\cR\colon R\ce{->[$\kappa_{rs}$]}S\ce{<=>[$\gamma_s$][$\beta_s$]}0\ce{<=>[$\beta_i$][$\gamma_i$]}I\ce{->[$\kappa_{ir}$]}R\ce{->[$\gamma_r$]}0\quad S+I\ce{->[$\kappa_{si}$]}2I,
\end{align*}
where respectively, species $S$ stands for the susceptible, $I$ the infected, and $R$ the recovered compartment of individuals; $\beta_s$ and $\beta_i$ are the  birth rates of susceptibles and infected individuals; $\gamma_s$, $\gamma_i$, and $\gamma_r$ are the death rates of susceptibles, infected individuals, as well as recovered individuals; $\kappa_{si}$ is the infection rate, $\kappa_{ir}$ is the recovery rate, and $\kappa_{rs}$ is the rate that a recovered individual loses immunity and immediately becomes susceptible again. {It is noteworthy that the growth rate is assumed to be constant as a consequence of demographics taken into consideration in epidemic models \cite{H20}. 
It follows from Theorem~\ref{thm:endotactic-regularity}(a) that this SIR{S} model as a mass-action system is regular independently of the choice of the positive rate constants. 
\subsection*{Volterra model}
Consider the following reaction network based on the classical deterministic Volterra competition system \cite{V31} among $d$ ecological species:
\[
\cR\colon 0\ce{<=>[$\alpha_i$][$\delta_i$]}S_i\ce{->[$\beta_i$]} 2S_i\quad S_i+S_j\ce{->[$\gamma_{ij}$]}S_i,\quad i,j=1,\ldots,d,
\]
where for species $S_i$, $\alpha_i$ is the immigration rate, $\beta_i$ the reproduction rate, and $\delta_i$ the death/ emigration rate; $\gamma_{ij}$ is the competition rate between species $S_i$  and $S_j$ due to overpopulation. Again by Theorem~\ref{thm:endotactic-regularity}(a), the Volterra model is a regular mass-action system.

\subsection*{Generalized Togashi-Kaneko model}
\! Consider the generalized Togashi-Kaneko (TK) model \cite{BKW20}:
\[
\cR\colon S_i+S_j\ce{->[$\kappa_{ij}$]}2S_j\quad S_i\ce{<=>[$\kappa_i$][$\kappa_{-i}$]}0 ,\quad i,j\in[d],\ i\neq j
\]
where the reaction rate constants are \emph{non-negative}, and reactions are removed whenever the  corresponding reaction rate constants are zero. It {has recently} been discovered that this model undergoes  \emph{discreteness-induced transitions} \cite{BKW20}. Observe that $\cR$ is binary, and Theorem~\ref{thm:endotactic-regularity}(a) applies to this model, yielding that this generalized TK model as a mass-action system is regular.

\subsection*{\textit{Glycolysis system}}
Consider the glycolysis chemical system \cite{H67,FH74}
\[
\cR\colon 0\ce{->[$\kappa_1$]} S_1\ \ S_1+S_2\ce{->[$\kappa_2$]} S_3+S_4\ \  S_3\ce{->[$\kappa_3$]} S_2\ce{<=>[$\kappa_4$][$\kappa_5$]}S_4+S_5\ \ S_4+S_6\ce{->[$\kappa_6$]} S_7\ce{->[$\kappa_7$]} S_6
\]
Since $\cR$ is binary, it follows from Theorem~\ref{thm:endotactic-regularity}(a) that $\cR$ is a regular mass-action system.

\subsection*{\textit{Enzyme reactions in a semi-open reactor}}
Consider an open enzymatic reaction network 
\cite[Section~4.2.2]{F19}:
\[
\cR\colon \begin{cases}
S_1+S_3 \ce{<=>[$\kappa_1$][$\kappa_2$]} S_4\quad S_2+S_3 \ce{<=>[$\kappa_3$][$\kappa_4$]} S_5\\
S_2+S_4 \ce{<=>[$\kappa_5$][$\kappa_6$]} S_6 \ce{<=>[$\kappa_7$][$\kappa_8$]}S_1+S_5\\
S_6 \ce{->[$\kappa_9$]} S_3+S_7\\
S_1\ce{<=>[$\kappa_{10}$][$\kappa_{11}$]}0\ce{<=>[$\kappa_{12}$][$\kappa_{13}$]}S_2\quad S_7\ce{->[$\kappa_{14}$]}0
\end{cases}
\]
where $S_1$ and $S_2$ are substrates, and $S_3$ stands for an enzyme. In the first step (reactions in the first line), both substrates $S_1$ and $S_2$ bind to the enzyme $S_3$ in respective binding sites to form two intermediate species $S_4$ and $S_5$. In the second step (reactions in the second line), $S_1$ binds reversibly to $S_5$, and similarly, $S_2$ binds reversibly to $S_4$. Once both substrates are bound to the enzyme in the form of the new intermediate species $S_6$, the product $S_7$ is formed (the reaction in the third line). The reactions in the last line refer to the influx of the substrates as well as the  efflux of the substrates and the product.

Again one can apply Theorem~\ref{thm:endotactic-regularity}(a) to yield $\cR$  is a regular mass-action system.

\subsection*{\textit{Autocatalytic bromate-bromide-cerium (III)\! reaction}}

It is well-known that oscillations of concentrations of species were observed in {the} Belousov-\\
\noindent Zhabotinskii reaction \cite{FKN72}. Consider the autocatalytic bromate-bromide-cerium (III) reaction process whose mechanism, as a simplification of the full mechanism given in \cite{FN74}, can be formulated as the following \emph{closed} reaction network:
\[
\cR\colon \hspace{-.2cm}\begin{cases}
\hspace{-.3cm}&S_1+S_2+2S_3\ce{->[$\kappa_1$]} S_4+S_5\ \  S_2+S_3+S_4\ce{->[$\kappa_2$]}2S_5\\
\hspace{-.5cm}& S_1+S_3+S_4\ce{->[$\kappa_3$]}2S_6+S_7\ \ S_3+S_6+S_8\ce{->[$\kappa_4$]}S_4+S_9\ \ 
2S_4\ce{->[$\kappa_5$]}S_1+S_3+S_5
\end{cases}
\]
where $S_1=\text{BrO}_3^-$, $S_2=\text{Br}^-$, $S_3=\text{H}^+$, $S_4=\text{HBrO}_2$, $S_5=\text{HOBr}$, $S_6=\text{BrO}_2^-$, $S_7=\text{H}_2\text{O}$, $S_8=\text{Ce}^{3+}$, and $S_9=\text{Ce}^{4+}$. It is easy to verify that $(2,2,1,3,3,2,2,1,1)$ is {a} conservation law.

To explain the oscillations of molar concentrations of species in the experiment, a reduced chemical  mechanism leading to an \emph{open}  deterministic reaction system was proposed in \cite{GF77}, where molecule counts of all other species but $S_2$, $S_4$, and $S_9$ are regarded as  constants. Hence the reduced open reaction network is given {as}:
\[
\widetilde{\cR}\colon
S_2\ce{->[$\widetilde{\kappa}_1$]} S_4\ce{->[$\widetilde{\kappa}_3$]}0\ce{<-[$\widetilde{\kappa}_5$]}2S_4\quad
S_2+S_4\ce{->[$\widetilde{\kappa}_2$]}0\ce{->[$\widetilde{\kappa}_4$]}S_4+S_9
\]
where the reaction rate constants with a tilde symbol are the new constants based on the assumption for the reduced reaction network.

It is readily verified that Theorem~\ref{thm:endotactic-regularity}(a) applies to $\widetilde{\cR}$, and hence $\widetilde{\cR}$ is a regular mass-action system.

\begin{figure}[h]
\begin{center}
\includegraphics[width=1\textwidth]{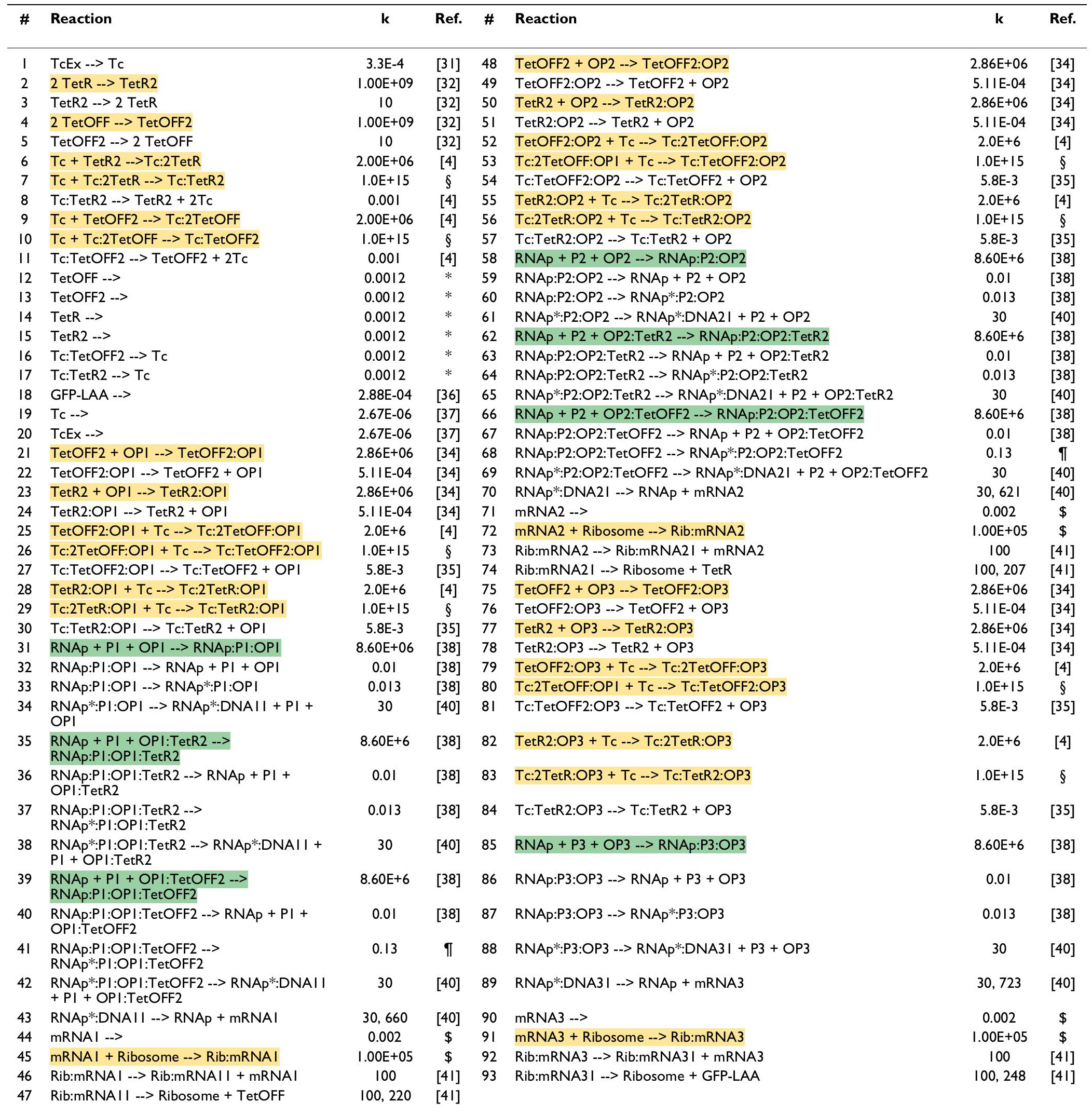}
\caption{Adapted from \cite[Table~1]{SK07}. The figure represents a synthetic tetracycline-inducible regulatory network of 93 reactions.  Bimolecular reactions are colored yellow and  termolecular reactions are colored  green. To read this table with more details, the reader is referred to \cite{SK07}.}
\end{center}\label{fig:93-reactions}
\end{figure}

\subsection*{Circadian clock}
Consider the circadian clock network $\cR$ of 9 species and 16 reactions \cite{VKBL02}:
\[
\scalebox{0.86}{
$
\cR\colon \begin{cases}
A\ce{->[$\delta_{A}$]}0\ce{<-[$\delta_{M_A}$]}M_A\ce{<-[$\beta_A$]} M_A + A\ \  D_A\ce{->[$\alpha_A$]}D_A+M_A\ \ D_A+A\ce{<=>[$\gamma_A$][$\theta_A$]}D_A'\ce{->[$\alpha_A'$]} D_A'+M_A\\
R\ce{->[$\delta_{R}$]}0\ce{<-[$\delta_{M_R}$]}M_R\ce{<-[$\beta_R$]} M_R + R\ \ D_R\ce{->[$\alpha_R$]}D_R+M_R\ \ D_R+R\ce{<=>[$\gamma_A$][$\theta_R$]}D_R'\ce{->[$\alpha_R'$]} D_R'+M_R\\
A+R\ce{->[$\gamma_C$]}C\ce{->[$\delta_A$]}R
\end{cases}
$
}
\]

where $A$ is the activator protein, $R$ {is} the repressor protein, and $C$ is the protein formed by binding $A$ and $R$; $D_A$ and $D_A'$ are the activator promoters with and without $A$ bound to them, respectively;  $D_R$ and $D_R'$ are the activator promoters  with and without $R$ bound to them, respectively; $M_A$ and $M_R$ are mRNA of $A$ and $R$, respectively.

It is straightforward to verify that Theorem~\ref{thm:endotactic-regularity}(a) applies, and thus the circadian clock network is a regular mass-action system.

In the following, we provide models which are \emph{not} binary.

\subsection*{\textit{Brusselator reaction network}}
Let us revisit the Brusselator reaction network:
\[
\cR\colon 0\ce{<=>}S_1\ce{->} S_2\quad 2S_1+S_2\ce{->} 3S_1
\]
By Theorem~\ref{thm:1st-order-regularity}(C1), the Brusselator mass-action system is regular.

\subsection*{\textit{An open peroxidase-oxidase (PO) reaction system}}

Consider the following peroxidase-oxidase (PO) reaction in an open system whose chemical mechanism is given by the following termolecular reaction network \cite{OD78}:
\[
\cR\colon \begin{cases}
S_1+S_2+S_3 \ce{->[$\kappa_1$]} 2 S_3\ce{->[$\kappa_2$]} 2 S_4\quad
 S_1+S_2+S_4 \ce{->[$\kappa_3$]} 2 S_3\\
0\ce{->[$\kappa_4$]} S_3 \ce{->[$\kappa_5$]} S_5\quad
S_4 \ce{->[$\kappa_6$]} S_6\quad
S_1 \ce{<=>[$\kappa_7$][$\kappa_8$]}S_7 \quad
S_8 \ce{->[$\kappa_9$]} S_2
\end{cases}
\]
Theorem~\ref{thm:1st-order-regularity}(C1) yields that $\cR$ is a regular mass-action system.

\subsection*{Synthetic tetracycline-inducible regulatory network}

Let us now continue to consider a more involved regulatory network $\cR$\---a synthetic tetracycline-inducible regulatory network \cite{SK07} given in Figure~\ref{fig:93-reactions}. 

In \cite{SK07}, not all reactions of the reaction network are  endowed with mass-action kinetics, as some  with $\S$ in the reference column of the table are adjusted for fast-reacting intermediates. Even if it might be possible to transform all reactions of non-mass-action kinetics into multi-step reactions of mass-action kinetics, we here prefer to save the effort while doing a bit of cheating by \emph{assuming mass-action kinetics for all reactions of no more than one molecularity}. 

One can carefully check the table given in Figure~\ref{fig:93-reactions} that it is a termolecular reaction network  consisting of 63 species and 93 distinct chemical reactions. In total, there are 57 unimolecular reactions, 27 bimolecular reactions, and 9 termolecular reactions. It is readily observed that all bimolecular and termolecular reactions are \emph{binding} reactions and thus the total molecule counts do not increase in these reactions. By Theorem~\ref{thm:1st-order-regularity}(C1), we know that $\cR$ is a regular reaction system in both senses.
\medskip

Below is an example where Theorem~\ref{thm:1st-order-regularity}(C1) fails while Theorem~\ref{thm:1st-order-regularity}(C2) and Theorem~\ref{thm:1st-order-regularity}(C3)  apply.

\subsection*{\textit{A generalized Lotka model}}
Consider the following model as a generalization of the classical Lotka model \cite{L20}:
\[\mathcal{R}\colon S_1\ce{->}2S_1\quad  mS_1+nS_2\ce{->}pS_1+rS_2\quad S_2\ce{->}0,\]
where $m>p$ and $n<r$ are all non-negative integers;  the first reaction stands for the birth of the predator ($S_2$), the last reaction stands for the natural death of the prey ($S_1$), and the reaction in the middle stands for predator-prey interaction. Here the general form of the second reaction accounts for the possibility that the unit mass of the prey and the predator  may not match \cite{MM07}: \emph{Either one can be much greater than the other}. For instance, at a macroscopic scale, a wolf ($S_2$) weighs heavier than a hare ($S_1$) and hence needs fewer molecules to conserve the total mass; while in contrast, at a microscopic scale, microbiome in the gut ($S_2$) may weigh much less than the host ($S_1$) and hence needs more molecules to conserve the total mass. In addition, the propensity function of the middle reaction may not obey the mass-action law. Indeed, $\cR\setminus\cR_*$ is conservative with a conservation law $(r-n,m-p)$. Hence the regularity of $\cR$ follows from Theorem~\ref{thm:1st-order-regularity}(C2). Besides, the molecule count of $S_1$ decreases in the unique reaction in $\cR\setminus\cR_*$, and hence Theorem~\ref{thm:1st-order-regularity}(C3) is also applicable. However, in a particular case where $r-n>m-p>0$ (e.g., microbiome-host model), then $\cR\setminus\cR_*=\{mS_1+nS_2\ce{->} pS_1+ rS_2\}$ contains a reaction where the total molecule counts increase, and hence Theorem~\ref{thm:1st-order-regularity}(C1) fails. 

\subsection*{\emph{A CRN in natural computing}} Consider the following reaction network $\cR$ which can deterministically compute a function \cite[Figure~2(b)]{CDS14}:
\[\begin{aligned}
\left.\begin{array}{l}X_1 \rightarrow X_1^g+X_1^h+X_1^\phi \\ X_2 \rightarrow X_2^g+X_2^h+X_2^\phi\end{array}\right\} & \text { make } 3 \text { copies of inputs for the } 3 \text { parallel computations }\\
\left.\begin{array}{l}X_1^g \rightarrow 2P^g \\ X_2^g \rightarrow C^g\end{array}\right\} & \text{ compute linear function } g(x_1, x_2)=2 x_1-x_2\\
X_2^h \rightarrow P^h & \text{ compute linear function } h(x_1, x_2)=x_2\\
\left.\begin{array}{l}X_1^\phi+F \rightarrow T \\ X_2^\phi+T \rightarrow F\end{array}\right\} & \text { compute semilinear predicate } \phi(x_1,x_2)=``x_1>x_2"?
\end{aligned}
\]
\[\begin{aligned} 
    \left.\begin{array}{l}T+Y+C^g \rightarrow T+\hat{C}^g \\ T+P^g \rightarrow T+Y+\hat{P}^g\end{array}\right\} & \begin{aligned} & \text { if predicate } \phi \text { true (ie } T \text {): } \\ & \text { every } P^g \text { increases } Y \text { and every } C^g \text { decreases } Y\end{aligned}
    \\
T+Y+\hat{P}^h \rightarrow T+P^h & \text{ undo the effect of } F
\\
    \left.\begin{array}{l}F+\hat{C}^g \rightarrow F+Y+C^g \\ F+Y+\hat{P}^g \rightarrow F+P^g\end{array}\right\} & \text { if predicate } \phi \text { false (ie } F \text{): undo the effect of } T\\ 
    F+P^h \rightarrow F+Y+\hat{P}^h & \text { every } P^h \text { increases } Y
    \end{aligned}
\]
where $X_i$ and their copies for $i=1,2$ stand for the input, $Y$ for output, and $P$ and their copies stand for ``produce'' and $C$ and their copies stand for ``consume''; $F$ and $T$ stand for ``false'' and ``true''. $\cR\setminus\cR_*$ consists of the last eight reactions which are bimolecular or termolecular, has a conservation law:\begin{equation}\label{eq:conservation_law}
\begin{split}&(X_1^g,X_2^g,X_1^h,X_2^h,X_1^{\phi},X_2^{\phi},C^g,\hat{C}^g,P^g,\hat{P}^g,P^h,\hat{P}^h,F,T,Y)\\
=& u=(1,1,1,1,1,2,1,2,2,1,2,1,2,1,1)
\end{split}
\end{equation}
Hence the regularity of $\cR$ as a mass-action system follows from Theorem~\ref{thm:1st-order-regularity}(C2). Nevertheless, it is easy to observe that the 9th, 11th, and 13th reactions are bimolecular with increasing total molecule counts; moreover, there exists no species of $\cR\setminus\cR_*$ that is non-increasing in all reactions of $\cR\setminus\cR_*$. Hence neither Theorem~\ref{thm:1st-order-regularity}(C1) nor Theorem~\ref{thm:1st-order-regularity}(C3) applies. 
\medskip

To close this section, we present the last model where Theorem~\ref{thm:1st-order-regularity}(C1)-(C2) fail while Theorem~\ref{thm:1st-order-regularity}(C3) applies.

\subsection*{\emph{An immunological model}}

Consider the following mass-action system as an immunological model \cite{HMM11}:
   \[
\cR\colon \begin{cases}
T_i\ce{<-[$\kappa_1$]}T+V{\ce{->[$\kappa_2$]}}2T+V\\
        M_i+V\ce{<-[$\kappa_3$]}M_i{\ce{<-[$\kappa_4$]}}M+V{\ce{->[$\kappa_5$]}}2M+V\\
       0\ce{<-[$\delta_1$]}T_i\ce{->[$\kappa_6$]}T_i+V\\
        M\ce{<=>[$\delta_1$][$s_2$]}0\ce{<=>[$s_1$][$\delta_1$]}T\\
        V\ce{->[$\delta_4$]}0\ce{<-[$\delta_5$]}M_i
    \end{cases}
    \]
Here $V$ stands for the virus; $T$ stands for new CD4+T cells and $T_i$ for infected CD4+T cells;  $M$ stands for macrophages and $M_i$ for infected macrophages. The total molecule counts increase in the four bimolecular reactions (with rate constants $\kappa_1,\ \kappa_2,\ \kappa_4,\ \kappa_5$), hence Theorem~\ref{thm:1st-order-regularity}(C1) fails to apply. Moreover, $\cR_*$ consisting of the four bimolecular reactions is not sub-conservative. Hence Theorem~\ref{thm:1st-order-regularity}(C2) fails as well. However, the molecule counts of species $V$ do not increase in any of the four bimolecular reactions, and it does decrease in the two reactions (with rate constants $\kappa_1$ and $\kappa_4$) that pair with the other two reactions  (with rate constants $\kappa_2$ and $\kappa_5$, respectively). By Theorem~\ref{thm:1st-order-regularity}(C3), $\cR$ is a regular mass-action system.

\section{Discussion}

In the light of (C1) in Theorem~\ref{thm:1st-order-regularity} (i.e., all reactions where the total molecule counts increase consist a unimolecular reaction network), while not necessary as evidenced by the generalized Lotka model, the reaction system may be irregular if there exists a reaction of higher than first-order where the total molecule counts increase.

\begin{example}\label{ex:run-away}
  Consider the following single-species population model including bursty pair
reproduction \cite{BA16}:
    \[S\ce{->}0\quad 2S\ce{->}rS,\quad r=3,\ldots,M,\]
where $M>3$ is the truncation number for the sequence of reactions accounting for the bursty pair reproduction ($M=\infty$ in \cite{BA16}; see also \cite{XHW23}). 
    This reaction network is a reduced model of one modeling bursty
viral production from infected cells \cite{PKP11} or reproduction of mammals involving a varying number of offspring \cite{BA16}. By \cite[Theorem~1]{XHW23}, this model as a stochastic mass-action system is explosive with probability one. Hence this model shows that a mass-action system can be irregular if there exists a higher (than first order) reaction that increases the total molecule counts of species.
\end{example}

Below we discuss the relevance of the main results to a concerned conjecture in CRNT.} A positive recurrence conjecture was proposed by Anderson and Kim in 2018 \cite{AK18}, which states that \emph{every weakly reversible stochastic mass-action system is positive recurrent} in each closed communicating class. Consistent progress has been made ever since \cite{AK18,AM18,ACK20,ACKN20,XHW23,WX20,AK23,LR23,KK24,FKY24}.
A classical tool to attack this problem relies on the so-called Foster-Lyapunov criterion  \cite{MT93,MT09,MP14}. Primarily based on this criterion, continued efforts since \cite{AK18} have been made  towards this conjecture  \cite{AK18,ACK20,XHW23,WX20,AK23}. Likely in the light of the aforementioned practical considerations from biochemistry, this conjecture was first proved with diverse novel techniques such as tier-sequence, for \emph{bimolecular systems} under various additional assumptions \cite{AK18,ACK20,ACKN20}. Moreover, since \emph{complex-balancing} \cite{H72} is also a pleasant property in biochemistry, it is interesting to know if a complex-balanced stochastic reaction system is  also positive recurrent  \cite{ACKK18}. It has been  proved that every complex-balanced stochastic mass-action system is non-explosive \cite[Theorem~3]{ACKK18}, and hence is positive recurrent \cite[Theorem~3.5.3]{N98} in the light of the light tailed  Poisson stationary distribution \cite[Theorem~4.1]{ACK10}.  This further implies that every weakly reversible \emph{deficiency zero} stochastic mass-action system is also positive recurrent on the grounds that it is complex-balanced \cite{H72}. As a consequence, this conjecture was also confirmed for \emph{unimolecular} systems  (i.e., every non-zero reactant is a copy of one species), since every unimolecular weakly reversible reaction network is of deficiency zero \cite{ACK10}.

Later, this positive recurrence conjecture was proved in one-dimension in \cite{WX20}, the arguably simplest scenario, based on a handy criterion recently established in \cite{XHW23} for positive recurrence of one-dimensional CTMCs of polynomial transition rates.

Nevertheless, in spite of all these efforts, the conjecture in the bimolecular case has yet to be closed. Theorem~\ref{thm:reducing-positive-recurrence-to-existence-of-SD} makes it possible to reduce this problem of proving positive recurrence for bimolecular weakly reversible stochastic mass-action systems to finding a stationary distribution for these systems, as illustrated by Example~\ref{ex:illustration-with-existence-of-SD}.
Hence, rather than constructing Lyapunov functions for positive recurrence in the sense of \cite[condition (CD2)]{MT93}, one may count on alternative techniques to prove the existence of stationary distributions \cite{ACK10,CJ18,HM19,AN19,H21,HKASK21,KTSB21,PH22,
CJ22,HH23,HHKK23,HWXia23,HWX23,WX23} to show positive recurrence for bimolecular systems.

It is noteworthy that although} linear functions turn out to be Lyapunov functions for the regularity of second-order weakly reversible  stochastic mass-action systems, in general they fail to work for \emph{termolecular} weakly reversible stochastic mass-action systems, as evidenced by Example~\ref{ex:linear-function-fails-for-3rd-order-WR}. This implies that even if every weakly reversible stochastic mass-action system were indeed  regular \cite{AK18} regardless of the order, it could be challenging to construct a \emph{simple} Lyapunov function to prove the regularity of weakly reversible  mass-action systems of a higher order. 

\section*{Acknowledgments}
The author appreciates the valuable comments from the editor as well as three anonymous referees which strengthen the results and improve the presentation of the manuscript. In particular, the author thanks an anonymous referee for an observation that eventually becomes Corollary~\ref{cor:one-species-interpretation}, and for proposing comparison of the regularity of deterministic reaction systems and that of stochastic reaction systems. The author thanks Anne Shiu and Daniele Cappelletti for their comments on the sub-conservativity which improves the statement of Theorem~\ref{thm:1st-order-regularity}. The author also thanks the Conference on Chemical Reaction Networks in Hawai\textquoteleft i 2026 (CRNT2026) which provides the opportunity for the author to report this work that triggered the above comments from Anne Shiu and Daniele Cappelletti as well as to learn the immunological model from Oskar Henriksson's tutorial lecture. 

\bibliographystyle{plain}
{\small\bibliography{references}}

@article{W72,
  title={{Dissipative dynamical systems part I: General theory}},
  author={Willems, Jan C},
  journal={Archive for Rational Mechanics and Analysis},
  volume={45},
  number={5},
  pages={321--351},
  year={1972},
  publisher={Springer}
}

@article{PKP11,
  title={{Stochastic theory of early viral infection: continuous versus burst production of virions}},
  author={Pearson, John E and Krapivsky, Paul and Perelson, Alan S},
  journal={PLoS Computational Biology},
  volume={7},
  number={2},
  pages={e1001058},
  year={2011},
  publisher={Public Library of Science San Francisco, USA}
}

@article{CDS14,
  title={{Deterministic function computation with chemical reaction networks}},
  author={Chen, Ho-Lin and Doty, David and Soloveichik, David},
  journal={Natural Computing},
  volume={13},
  number={4},
  pages={517--534},
  year={2014},
  publisher={Springer}
}

@article{ADE18b,
  title={On the geometry of chemical reaction networks: Lyapunov function and large deviations},
  author={Agazzi, Andrea and Dembo, Amir and Eckmann, J-P},
  journal={Journal of Statistical Physics},
  volume={172},
  number={2},
  pages={321--352},
  year={2018},
  publisher={Springer}
}

@book{MM07,
  title={Theoretical Ecology: Principles and Applications},
  author={May, Robert and McLean, Angela R},
  year={2007},
  publisher={OUP Oxford}
}

@article{H20,
  title={{The mathematics of infectious diseases}},
  author={Hethcote, H.W.},
  journal={SIAM Review},
  volume={42},
  number={4},
  pages={599--653},
  year={2000},
  publisher={SIAM}
}

@article{BA16,
	Author = {Be'er, S. and Assaf, M.},
	Journal = {J. Stat. Mech--Theory E.},
	Pages = {113501},
	Title = {{Rare events in stochastic populations under bursty reproduction}},
	Volume = {},
	Year = {2016}}

@book{N98,
  title={{Markov Chains}},
  author={Norris, James R},
  series={{Cambridge Series in Statistical and Probabilistic Mathematics}},
  number={2},
  year={1998},
  publisher={Cambridge University Press}
}

@article{K70,
  title={{Solutions of ordinary differential equations as limits of pure jump Markov processes}},
  author={Kurtz, Thomas G},
  journal={Journal of Applied Probability},
  volume={7},
  number={1},
  pages={49--58},
  year={1970},
  publisher={Cambridge University Press}
}

@article{K71,
  title={{Limit theorems for sequences of jump Markov processes approximating ordinary differential processes}},
  author={Kurtz, Thomas G},
  journal={Journal of Applied Probability},
  volume={8},
  number={2},
  pages={344--356},
  year={1971},
  publisher={Cambridge University Press}
}

@article{ACK10,
  title={{Product-form stationary distributions for deficiency zero chemical reaction networks}},
  author={Anderson, David F and Craciun, Gheorghe and Kurtz, Thomas G},
  journal={Bulletin of Mathematical Biology},
  volume={72},
  pages={1947--1970},
  year={2010},
  publisher={Springer}
}

@article{PKC12,
  title={{Global injectivity and multiple equilibria in uni-and bi-molecular reaction networks}},
  author={Pantea, Casian and Koeppl, Heinz and Craciun, Gheorghe},
  journal={Discrete Contin. Dyn. Syst. Ser. B},
  volume={17},
  number={6},
  pages={2153--2170},
  year={2012}
}

@article{F87,
  title={{Chemical reaction network structure and the stability of complex isothermal reactors-I. The deficiency zero and deficiency one theorems}},
  author={Feinberg, Martin},
  journal={Chemical Engineering Science},
  volume={42},
  number={10},
  pages={2229--2268},
  year={1987},
  publisher={Elsevier}
}

@article{ACK20,
  title={Stochastically modeled weakly reversible reaction networks with a single linkage class},
  author={Anderson, David F and Cappelletti, Daniele and Kim, Jinsu},
  journal={Journal of Applied Probability},
  volume={57},
  number={3},
  pages={792--810},
  year={2020},
  publisher={Cambridge University Press}
}

@article{LR23,
  title={{Scaling methods for stochastic chemical reaction networks}},
  author={Laurence, Lucie and Robert, Philippe},
  journal={Stochastic Processes and their Applications},
  volume = {194},
  doi={10.1016/j.spa.2025.104855},
  year={2026}
}

@article{AM18,
  title={{Seemingly stable chemical kinetics can be stable, marginally stable, or unstable}},
  author={Agazzi, Andrea and Mattingly, Jonathan C},
  journal={Communications in Mathematical Sciences},
  volume={18},
  number={6},
  pages={1605--1642},
  year={2018},
  publisher={Elsevier}
}

@article{GBK14,
  title={{A scalable computational framework for establishing long-term behavior of stochastic reaction networks}},
  author={Gupta, Ankit and Briat, Corentin and Khammash, Mustafa},
  journal={PLoS Computational Biology},
  volume={10},
  number={6},
  pages={e1003669},
  year={2014},
  publisher={Public Library of Science San Francisco, USA}
}

@article{KK24,
  title={{A path method for non-exponential ergodicity of Markov chains and its application for chemical reaction systems}},
  author={Kim, Minjun and Kim, Jinsu},
  journal={Journal of Statistical Physics},
  volume={192},
  number={74},
  doi={10.1007/s10955-025-03453-6},
  year={2025},
  publisher={Springer}
}

@article{AK18,
  title={{Some network conditions for positive recurrence of stochastically modeled reaction networks}},
  author={Anderson, David F and Kim, Jinsu},
  journal={SIAM Journal on Applied Mathematics},
  volume={78},
  number={5},
  pages={2692--2713},
  year={2018},
  publisher={SIAM}
}

@article{AK23,
  title={{Mixing times for two classes of stochastically modeled reaction networks}},
  author={Anderson, David F and Kim, Jinsu},
  journal={Mathematical Biosciences and Engineering},
  volume={20},
  number={3},
  pages={4690-4713},
  year={2023},
  publisher={AIMS}
}

@article{GMS14,
  title={{A geometric approach to the global attractor conjecture}},
  author={Gopalkrishnan, Manoj and Miller, Ezra and Shiu, Anne},
  journal={SIAM Journal on Applied Dynamical Systems},
  volume={13},
  number={2},
  pages={758--797},
  year={2014},
  publisher={SIAM}
}

@article{PCK14,
  title={{Dynamical properties of discrete reaction networks}},
  author={Paulev{\'e}, Lo{\"\i}c and Craciun, Gheorghe and Koeppl, Heinz},
  journal={Journal of Mathematical Biology},
  volume={69},
  pages={55--72},
  year={2014},
  publisher={Springer}
}

@article{GF77,
  title={{Three steady state situation in an open chemical reaction system. I}},
  author={Geiseler, W and F{\"o}llner, H H},
  journal={Biophysical Chemistry},
  volume={6},
  number={2},
  pages={107--115},
  year={1977},
  publisher={Elsevier}
}

@article{OD78,
  title={{Oscillatory kinetics of the peroxidase-oxidase reaction in an open system. Experimental and theoretical studies}},
  author={Olsen, Lars F and Degn, Hans},
  journal={Biochimica et Biophysica Acta (BBA)-Enzymology},
  volume={523},
  number={2},
  pages={321--334},
  year={1978},
  publisher={Elsevier}
}

@article{BKW20,
  title={{Stationary distributions of systems with discreteness-induced transitions}},
  author={Bibbona, Enrico and Kim, Jinsu and Wiuf, Carsten},
  journal={Journal of The Royal Society Interface},
  volume={17},
  number={168},
  pages={20200243},
  year={2020},
  publisher={The Royal Society}
}

@article{HKASK21,
  title={Derivation of stationary distributions of biochemical reaction networks via structure transformation},
  author={Hong, Hyukpyo and Kim, Jinsu and Ali Al-Radhawi, M and Sontag, Eduardo D and Kim, Jae K},
  journal={Communications Biology},
  volume={4},
  number={1},
  pages={620},
  year={2021},
  publisher={Nature Publishing Group UK London}
}

@article{HHKK23,
  title={{Computational translation framework identifies biochemical reaction networks with special topologies and their long-term dynamics}},
  author={Hong, Hyukpyo and Hernandez, Bryan S and Kim, Jinsu and Kim, Jae K},
  journal={SIAM Journal on Applied Mathematics},
  volume={83},
  number={3},
  pages={1025--1048},
  year={2023},
  publisher={SIAM}
}

@article{MFS22,
  title={{What makes a reaction network ``chemical''?}},
  author={M{\"u}ller, Stefan and Flamm, Christoph and Stadler, Peter F},
  journal={Journal of Cheminformatics},
  volume={14},
  number={63},
  year={2022},
  publisher={Springer}
}

@article{VKBL02,
  title={Mechanisms of noise-resistance in genetic oscillators},
  author={Vilar, Jos{\'e} MG and Kueh, Hao Y and Barkai, Naama and Leibler, Stanislas},
  journal={Proceedings of the National Academy of Sciences},
  volume={99},
  number={9},
  pages={5988--5992},
  year={2002},
  publisher={National Acad Sciences}
}

@article{HH23,
  title={{Squeezing stationary distributions of stochastic chemical reaction systems}},
  author={Hirono, Yuji and Hanai, Ryo},
  journal={Journal of Statistical Physics},
  volume={190},
  number={4},
  pages={86},
  year={2023},
  publisher={Springer}
}

@article{FH74,
  title={{Dynamics of open chemical systems and the algebraic structure of the underlying reaction network}},
  author={Feinberg, Martin and Horn, Friedrich JM},
  journal={Chemical Engineering Science},
  volume={29},
  number={3},
  pages={775--787},
  year={1974},
  publisher={Elsevier}
}

@article{H67,
  title={{The theory of oscillating reactions-kinetics symposium}},
  author={Higgins, Joseph},
  journal={Industrial \& Engineering Chemistry},
  volume={59},
  number={5},
  pages={18--62},
  year={1967},
  publisher={ACS Publications}
}

@article{ACKN20,
  title={{Tier structure of strongly endotactic reaction networks}},
  author={Anderson, David F and Cappelletti, Daniele and Kim, Jinsu and Nguyen, Tung D},
  journal={Stochastic Processes and their Applications},
  volume={130},
  number={12},
  pages={7218--7259},
  year={2020},
  publisher={Elsevier}
}

@book{F19,
  title={{Foundations of Chemical Reaction Network Theory}},
  author={Feinberg, Martin},
  year={2019},
  publisher={Springer Science \& Business Media}
}

@article{CNP13,
  title={{Persistence and permanence of mass-action and power-law dynamical systems}},
  author={Craciun, Gheorghe and Nazarov, Fedor and Pantea, Casian},
  journal={SIAM Journal on Applied Mathematics},
  volume={73},
  number={1},
  pages={305--329},
  year={2013},
  publisher={SIAM}
}

@article{P12,
  title={On the persistence and global stability of mass-action systems},
  author={Pantea, Casian},
  journal={SIAM Journal on Mathematical Analysis},
  volume={44},
  number={3},
  pages={1636--1673},
  year={2012},
  publisher={SIAM}
}

@article{H72,
  title={{Necessary and sufficient conditions for complex balancing in chemical kinetics}},
  author={Horn, Fritz},
  journal={Archive for Rational Mechanics and Analysis},
  volume={49},
number={3},
  pages={172--186},
  year={1972},
  publisher={Springer}
}

@article{H21,
  title={Stationary distributions via decomposition of stochastic reaction networks},
  author={Hoessly, Linard},
  journal={Journal of Mathematical Biology},
  volume={82},
  number={7},
  pages={67},
  year={2021},
  publisher={Springer}
}

@book{MT09,
  title={{Markov Chains and Stochastic Stability}},
  author={Meyn, Sean P and Tweedie, Richard L},
  year={2009},
  publisher={Springer Science \& Business Media}
}

@article{MT93,
  title={{Stability of Markovian processes III: Foster--Lyapunov criteria for continuous-time processes}},
  author={Meyn, Sean P and Tweedie, Richard L},
  journal={Advances in Applied Probability},
  volume={25},
  number={3},
  pages={518--548},
  year={1993},
  publisher={Cambridge University Press}
}

@book{EK09,
  title={{Markov Processes: Characterization and Convergence}},
  author={Ethier, Stewart N and Kurtz, Thomas G},
  year={2009},
  publisher={John Wiley \& Sons}
}

@article{AN19,
  title={{Results on stochastic reaction networks with non-mass action kinetics}},
  author={Anderson, David F and Nguyen, Tung D},
  journal={Mathematical Biosciences and Engineering},
  volume={16},
  pages={2118--2140},
  year={2019}
}

@article{KTSB21,
  title={{Stationary distributions of continuous-time Markov chains: a review of theory and truncation-based approximations}},
  author={Kuntz, Juan and Thomas, Philipp and Stan, Guy-Bart and Barahona, Mauricio},
  journal={SIAM Review},
  volume={63},
  number={1},
  pages={3--64},
  year={2021},
  publisher={SIAM}
}

@incollection{AK11,
  title={{Continuous time Markov chain models for chemical reaction networks}},
  author={Anderson, David F and Kurtz, Thomas G},
  booktitle={Design and analysis of biomolecular circuits: engineering approaches to systems and synthetic biology},
  pages={3--42},
  year={2011},
  publisher={Springer}
}

@article{L20,
  title={{Undamped oscillations derived from the law of mass action}},
  author={Lotka, Alfred J},
  journal={{Journal of the American Chemical Society}},
  volume={42},
  number={8},
  pages={1595--1599},
  year={1920},
  publisher={ACS Publications}
}

@book{V31,
  title={Th{\'e}orie math{\'e}matique de la lutte pour la vie},
  author={Volterra, Vito},
  year={1931},
  publisher={Gauthiers-Villars}
}

@article{FMKT06,
  title={Modeling the kinetics of bimolecular reactions},
  author={Fern{\'a}ndez-Ramos, Antonio and Miller, James A and Klippenstein, Stephen J and Truhlar, Donald G},
  journal={Chemical Reviews},
  volume={106},
  number={11},
  pages={4518--4584},
  year={2006},
  publisher={ACS Publications}
}

@book{C06,
  title={{General Chemistry: The Essenstial Concepts}},
  author={Chang, Raymond},
  year={2006},
  edition={5th},
  publisher={McGraw-Hill}
}

@article{KK88,
  title={{Kinetics of bimolecular reactions in condensed media: critical phenomena and microscopic self-organisation}},
  author={Kuzovkov, V. and Kotomin, E.},
  journal={Reports on Progress in Physics},
  volume={51},
  number={12},
  pages={1479},
  year={1988},
  publisher={IOP Publishing}
}

@article{BFP10,
  title={{Stochastic Turing patterns in the Brusselator model}},
  author={Biancalani, Tommaso and Fanelli, Duccio and Di Patti, Francesca},
  journal={Physical Review E-Statistical, Nonlinear, and Soft Matter Physics},
  volume={81},
  number={4},
  pages={046215},
  year={2010},
  publisher={APS}
}

@article{JBP02,
  title={{Futile cycles revisited: a Markov chain model of simultaneous glycolysis and gluconeogenesis}},
  author={Jones, M E and Berry, M N and Phillips, J W},
  journal={Journal of Theoretical Biology},
  volume={217},
  number={4},
  pages={509--523},
  year={2002},
  publisher={Elsevier}
}

@article{LH07,
  title={Stochastic fluctuations in metabolic pathways},
  author={Levine, Erel and Hwa, Terence},
  journal={{Proceedings of the National Academy of Sciences}},
  volume={104},
  number={22},
  pages={9224--9229},
  year={2007},
  publisher={National Acad Sciences}
}

@article{GMK17,
  title={{Dynamic disorder in simple enzymatic reactions induces stochastic amplification of substrate}},
  author={Gupta, Ankit and Milias-Argeitis, Andreas and Khammash, Mustafa},
  journal={Journal of The Royal Society Interface},
  volume={14},
  number={132},
  pages={20170311},
  year={2017},
  publisher={The Royal Society}
}

@article{SK07,
  title={{Synthetic tetracycline-inducible regulatory networks: computer-aided design of dynamic phenotypes}},
  author={Sotiropoulos, Vassilios and Kaznessis, Yiannis N},
  journal={BMC Systems Biology},
  volume={1},
  pages={1--18},
  year={2007},
  publisher={Springer}
}

@article{M67,
  title={{Stochastic approach to chemical kinetics}},
  author={McQuarrie, Donald A},
  journal={{Journal of Applied Probability}},
  volume={4},
  number={3},
  pages={413--478},
  year={1967},
  publisher={Cambridge University Press}
}

@article{WX20,
  title={Classification and threshold dynamics of stochastic reaction networks},
  author={Wiuf, Carsten and Xu, Chuang},
  journal={arXiv:2012.07954},
  year={2020}
}

@article{HWX23,
  title={{Stationary measures of continuous-time Markov chains with applications to stochastic reaction networks}},
  author={Hansen, Mads Christian and Wiuf, Carsten and Xu, Chuang},
  journal={Advances in Applied Probability},
  volume={57},
  number={3},
  pages={1031--1067},
  year={2025},
  publisher={Cambridge University Press}
}

@article{WX23,
  title={Any stochastic reaction network has a stationary measure},
  author={Wiuf, Carsten and Xu, Chuang},
  journal={arXiv:2312.07590},
  year={2023}
}

@article{HMM11,
  title={Towards modeling HIV long term behavior},
  author={Hernandez-Vargas, Esteban A and Mehta, Dhagash and Middleton, Richard H},
  journal={IFAC Proceedings Volumes},
  volume={44},
  number={1},
  pages={581--586},
  year={2011},
  publisher={Elsevier}
}

@article{ACKK18,
  title={{Non-explosivity of stochastically modeled reaction networks that are complex balanced}},
  author={Anderson, David F and Cappelletti, Daniele and Koyama, Masanori and Kurtz, Thomas G},
  journal={Bulletin of Mathematical Biology},
  volume={80},
  pages={2561--2579},
  year={2018},
  publisher={Springer}
}

@article{MP14,
  title={{Explosion, implosion, and moments of passage times for continuous-time Markov chains: a semimartingale approach}},
  author={Menshikov, Mikhail and Petritis, Dimitri},
  journal={Stochastic Processes and Their Applications},
  volume={124},
  number={7},
  pages={2388--2414},
  year={2014},
  publisher={Elsevier}
}

@article{XHW22,
  title={{Structural classification of continuous time Markov chains with applications}},
  author={Xu, Chuang and Hansen, Mads C and Wiuf, Carsten},
  journal={Stochastics},
  volume={94},
  number={7},
  pages={1003--1030},
  year={2022},
  publisher={Taylor \& Francis}
}

@article{XHW23,
  title={{Full classification of dynamics for one-dimensional continuous-time Markov chains with polynomial transition rates}},
  author={Xu, Chuang and Hansen, Mads C and Wiuf, Carsten},
  journal={Advances in Applied Probability},
  volume={55},
  number={1},
  pages={321--355},
  year={2023},
  publisher={Cambridge University Press}
}

@article{CJ18,
  title={{Graphically balanced equilibria and stationary measures of reaction networks}},
  author={Cappelletti, Daniele and Joshi, Badal},
  journal={SIAM Journal on Applied Dynamical Systems},
  volume={17},
  number={3},
  pages={2146--2175},
  year={2018},
  publisher={SIAM}
}

@article{CJ22,
  title={{Transition graph decomposition for complex balanced reaction networks with non-mass-action kinetics}},
  author={Cappelletti, Daniele and Joshi, Badal},
  journal={Mathematical Biosciences and Engineering},
  volume={19},
  number={8},
  pages={7649--7668},
  year={2022},
  publisher={AIMS}
}

@article{HM19,
  title={{Stationary distributions and condensation in autocatalytic reaction networks}},
  author={Hoessly, Linard and Mazza, Christian},
  journal={SIAM Journal on Applied Mathematics},
  volume={79},
  number={4},
  pages={1173--1196},
  year={2019},
  publisher={SIAM}
}

@article{PH22,
  title={{An algebraic approach to product-form stationary distributions for some reaction networks}},
  author={Pascual-Escudero, Beatriz and Hoessly, Linard},
  journal={SIAM Journal on Applied Dynamical Systems},
  volume={21},
  number={1},
  pages={588--615},
  year={2022},
  publisher={SIAM}
}

@article{X26,
  title={First-order endotactic reaction networks},
  author={Xu, Chuang},
  journal={Journal of Mathematical Biology},
  volume={92},
  number={74},
  doi={10.1007/s00285-026-02378-2},
  year={2026},
  publisher={Springer}
}

@article{FN74,
  title={{Oscillations in chemical systems. IV. Limit cycle behavior in a model of a real chemical reaction}},
  author={Field, Richard J and Noyes, Richard M},
  journal={The Journal of Chemical Physics},
  volume={60},
  number={5},
  pages={1877--1884},
  year={1974},
  publisher={AIP Publishing}
}

@article{FKN72,
  title={{Oscillations in chemical systems. II. Thorough analysis of temporal oscillation in the bromate-cerium-malonic acid system}},
  author={Field, Richard J and K\"{o}r\"{o}s, Endre and Noyes, Richard M},
  journal={Journal of the American Chemical Society},
  volume={94},
  number={25},
  pages={8649--8664},
  year={1972},
  publisher={ACS Publications}
}

@article{FKY24,
  title={{Boundary-induced slow mixing for Markov chains and its application to stochastic reaction networks}},
  author={Fan, Wai-Tong L. and Kim, Jinsu and Yuan, Chaojie},
  journal={arXiv:2407.12166},
  year={2024}
}

@article{HWXia23,
  title={{Complex balanced distributions for chemical reaction networks}},
  author={Hoessly, Linard and Wiuf, Carsten and Xia, Panqiu},
  journal={arXiv:2301.04091},
  year={2023}
}

\end{document}